%% file: supp-rank.tex
\documentclass{article}

\input{header}
\usepackage{wrapfig}
\usepackage{float}
\usepackage{afterpage}
\usepackage{accents}
\usepackage{bm}
\usepackage{centernot}
\usepackage[nottoc]{tocbibind}
\usepackage{tocloft}
\restylefloat{table}
\usepackage{etoolbox}
\usepackage{nicematrix}
\usepackage{caption}
\usepackage{xspace}
\usepackage{bbm}
\usepackage{cleveref}

\renewcommand{\epsilon}{\varepsilon}

\DeclareMathOperator{\signrank}{rank_\pm}
\DeclareMathOperator{\suprank}{rank_0}
\DeclareMathOperator{\margin}{\mathrm{mar}}
\DeclareMathOperator{\U}{U}
\DeclareMathOperator{\Rand}{R}

\DeclareMathOperator{\img}{image}
\newcommand{\Diag}{\mathrm{Diag}}

\DeclareMathOperator{\rank}{rank}

\DeclareMathOperator{\D}{\mathsf{D}}

\newcommand{\cplxclass}[1]{\mathsf{#1}}
\let\P\undefined\DeclareMathOperator{\P}{\cplxclass{P}}
\DeclareMathOperator{\BPP}{\cplxclass{BPP}}
\DeclareMathOperator{\UPP}{\cplxclass{UPP}}
\DeclareMathOperator{\SUPP}{\cplxclass{SUPP}}
\DeclareMathOperator{\coSUPP}{\cplxclass{coSUPP}}
\DeclareMathOperator{\RP}{\cplxclass{RP}}
\DeclareMathOperator{\coRP}{\cplxclass{coRP}}

\DeclareMathOperator{\SMAR}{\cplxclass{SMAR}}
\DeclareMathOperator{\coSMAR}{\cplxclass{coSMAR}}

\newcommand{\EQ}{\mathsf{EQ}}
\newcommand{\GEQ}{\mathsf{GT}}
\newcommand{\HD}[1]{\mathsf{HD}_{#1}}

\newcommand{\R}{\mathbb{R}}

\title{Sign-Rank of $k$-Hamming Distance is Constant}

\author{
Mika G\"o\"os \\ {\slshape\normalsize EPFL}
\and
Nathaniel Harms \\ {\slshape\normalsize EPFL}
\and
Valentin Imbach \\ {\slshape\normalsize EPFL}
\and
Dmitry Sokolov \\ {\slshape\normalsize EPFL}
}

\date{\vspace{0.5em}\today}

\begin{document}

\maketitle

\begin{abstract}\noindent
We prove that the sign-rank of the $k$-Hamming Distance matrix on $n$ bits is
$2^{O(k)}$, independent of the number of bits $n$. This strongly refutes the
conjecture of Hatami, Hatami, Pires, Tao, and Zhao ({\footnotesize RANDOM
2022}), and Hatami, Hosseini, and Meng ({\footnotesize STOC 2023}), repeated in
several other papers, that the sign-rank should depend on $n$.
This conjecture
would have qualitatively separated \emph{margin} from \emph{sign-rank} (or,
equivalently, bounded-error from unbounded-error randomized communication). In
fact, our technique gives constant sign-rank upper bounds for all matrices which
reduce to $k$-Hamming Distance, as well as large-margin matrices recently shown
to be \emph{irreducible} to $k$-Hamming Distance.
\end{abstract}

\vspace{2em}

{
\setcounter{tocdepth}{2} 
\tableofcontents
}
\thispagestyle{empty}
\setcounter{page}{0}
\newpage
\setcounter{page}{1}

\section{Introduction}

A boolean matrix $M \in \zo^{N \times N}$ can always be represented as a
point--halfspace arrangement. For example, in \cref{fig:identity}, the identity
and a lower triangular
matrix are represented with rows assigned to points
on the unit sphere and columns assigned to halfspaces with boundary through the
origin, such that point $x_i$ belongs to a halfspace $h_j$ if and only if the entry
$i,j$ of the matrix is 1.

\begin{figure}[h!]
    \begin{center}
        \input{pics/half-eq.tex}
    \end{center}
    \caption{\textcolor{Gray}{
            Sign-rank representations of the identity matrix and lower triangular matrix.
        }
    }
    \label{fig:identity}
\end{figure}

The smallest dimension $d$ in which $M$ can be
represented in this way, with rows assigned to points and columns assigned to
halfspaces through the origin, is called the \emph{sign-rank} of $M$, denoted
$\signrank(M)$ (also called \emph{dimension complexity} in learning theory). We
may equivalently define sign-rank of a boolean matrix~$M$ as the smallest rank
$d$ of a matrix $A \in \bR^{N \times N}$ which satisfies $M_\pm(i,j) =
\sign(A(i,j))$ for all $i,j$, where~$M_\pm(i,j) = 2M(i,j) - 1$ is a
representation of $M$ as a sign matrix $M_\pm \in \pmset^{N \times N}$.

In this paper we prove that the
\textsc{$k$-Hamming Distance} matrices have sign-rank at most $2^{O(k)}$. These
are the matrices $\HD{k}^n\colon \zo^n \times \zo^n \to \zo$ where
$\HD{k}^n(x,y) = 1$ if and only if the Hamming distance $\dist(x,y)$
between row index $x \in \zo^n$ and column index $y \in \zo^n$ is exactly $k$.
\begin{boxtheorem}
\label{thm:main}
    For all $n,k\in\bN$, $\signrank(\HD{k}^n) = 2^{O(k)}$.
\end{boxtheorem}
For $k=0$, $\HD{k}^n$ is the identity matrix in \cref{fig:identity}, but for $k
\geq 1$ our theorem improves on the best known (and trivial) bound of $\poly(n)$,
and refutes the conjecture that for some constant $k$ the sign-rank must depend
on $n$. This was a conjecture of Hatami, Hatami, Pires, Tao, and Zhao
\cite{HHPTZ22}, Hatami, Hosseini, and Meng \cite{HHM23}, the basis for a question
of Harms and Zamaraev \cite{HZ24}, and an open problem in
\cite{FHHH24,HH24,HR24,FGHH25}. The goal of this conjecture was to separate the
class of matrices with \emph{large margin}, \ie those which can be represented
as point--halfspace arrangements with a large (constant) margin between any
hyperplane and any point; from those of \emph{constant sign-rank}, \ie the ones
which can be represented as point--halfspace arrangements whose dimension is
independent of matrix size. This is a question of Linial, Mendelson, Schechtman,
and Shraibman \cite{LMSS07} with consequences for communication complexity,
learning theory, circuit complexity, distributed computing, privacy, and other
areas of computer science \cite{FKL+01, BES02, For02, LS09, FX14, BNS19, HWZ22,
HHPTZ22, HZ24, AN25}. We explain this question in more detail in
\cref{section:signrank-vs-margin}.

\paragraph{Generalization.}
Surprisingly, our simple technique applies not only to \textsc{$k$-Hamming
Distance} but to all large-margin matrices obtained from it by \emph{reductions} (\ie boolean combinations), as well as large-margin matrices which are
\emph{irreducible} to \textsc{$k$-Hamming Distance} \cite{FGHH25}. While new
large-margin matrices can be created by reductions, it is a well-known open
problem whether reductions preserve sign-rank~\cite{BMT21,HHPTZ22}, so
reductions to \textsc{$k$-Hamming Distance} are not handled \emph{a
priori} by \cref{thm:main}. Furthermore, \cite{FGHH25} proved that there exist
large-margin matrices which are \emph{irreducible} to \textsc{$k$-Hamming
Distance}. These can be obtained from \textsc{$k$-Hamming Distance} by
``distance-$r$ compositions''\!. \cite{FGHH25} observed that all large-margin
matrices known prior to their work could be obtained from \textsc{$k$-Hamming
Distance} by compositions and reductions. We show that none of these examples can
separate margin from sign-rank:

\begin{boxtheorem}[Generalization of \cref{thm:main}, Informal]
\label{thm:gen-main}
    Any boolean matrix $M$ that is obtained from \textsc{$k$-Hamming Distance} by reductions and
    ``compositions'' has $\signrank(M) = O(1)$.
\end{boxtheorem}
Subsequent to \cite{FGHH25} and concurrent with the present study, recent work
of Sherstov and Storozhenko \cite{SS24} introduces a new class of large-margin
matrices, which are now the only remaining candidates that we are aware of for
separating large margin from constant sign-rank.

The proofs of \cref{thm:main,thm:gen-main} involve new techniques (a focused
study of \emph{support-rank}) that we outline in \cref{sec:techniques}. These
techniques in turn lead us to define some useful new complexity classes, which
we describe in \cref{section:whatsupp} along with some open problems. For the
remainder of this introduction, we discuss the implications of our main results.

\subsection{Context and Consequences for Sign-Rank vs.\ Margin}
\label{section:signrank-vs-margin}

\paragraph{Sign-rank and margin.}
Given a boolean matrix $M$, we can ask to minimize the \emph{dimension} of its
point--halfspace representation, which leads to the definition of sign-rank. But
another way to optimize the point--halfspace representation is to ignore the
dimension and maximize the \emph{margin}, so that no point is too close to the
boundary of any halfspace. For a boolean matrix $M \in \zo^{N \times N}$, we
write
\[
  \margin(M) \define \sup_{u,v} \min_{i,j \in [N]} |\inn{u_i, v_j}|
\]
where the supremum is over all assignments $u, v \colon [N] \to \bR^d$ of the
rows and columns to unit vectors in any dimension $d$ such that the signed
matrix is $M_\pm(i,j) = \sign(\inn{u_i,v_j})$.

It is not well understood how these two types of representations relate to
one another. Small sign-rank does not imply large margin: the triangular matrix in \cref{fig:identity} has sign-rank 2 but small (sub-constant)
margin~\cite{BW16,Vio15,SA23}, while matrices of sign-rank 3 can have margin
(equivalently, discrepancy~\cite{LS09}) as small as $(\poly(N))^{-1}$
\cite{HHL20,ACHS24}. A basic open question is the converse:
\begin{boxquestion}[\cite{LMSS07,HHPTZ22}]
\label{question:intro-main}
Is there a function $\eta$ such that any boolean matrix $M$ satisfies $\signrank(M) \leq \eta(\margin(M)^{-1})$? That is, do matrices of large (constant) margin also have constant sign-rank?
\end{boxquestion}
Contrary to conjectures in earlier work, our \cref{thm:gen-main} shows that all large-margin
matrices covered in \cite{FGHH25} also have constant sign-rank. Sign-rank and margin are important in
several areas of computer science, leading to several equivalent formulations of
this question, for example:
\begin{itemize}[itemsep=0pt]
\item The margin determines the performance of the \emph{perceptron} algorithm. Are there hypothesis
classes of dimension (sign-rank) $\omega(1)$ that the perceptron algorithm can learn with only
$O(1)$ mistakes?
\item Via a relationship to randomized communication complexity, hypothesis classes $\cH$ with
constant margin are exactly those which are PAC learnable under pure differential privacy \cite{FX14,BNS19}.  Is any
super-constant dimensional problem learnable under pure differential privacy?
\item Is every hereditary graph class with constant size \emph{adjacency sketches}
(\eg \cite{FK09,Har20,HWZ22,EHK22,NP24,AN25}) a point--halfspace incidence graph in constant dimension?
\end{itemize}
We next focus on communication complexity, where, as we explain below, the equivalent question is:
\begin{boxquestion}
Is there any communication problem with \emph{constant} bounded-error randomized cost, but
\emph{super-constant} unbounded-error cost?
\end{boxquestion}

\paragraph{Sign-rank in communication.}
One of the main goals in communication complexity is to understand the power of
randomness. When allowing randomness in a communication protocol, there are a
few choices we can make about what to demand from our protocol:
\begin{enumerate}[itemsep=0pt]
\item Is the source of randomness \emph{public-coin} (both parties share the source of randomness),
or \emph{private-coin} (each party has their own source of randomness that the other doesn't see)?
\item Should our protocol have \emph{bounded error} (the probability of error is at most, say $1/4$)
or are we satisfied with \emph{unbounded error} (the probability of error is strictly less than
$1/2$)?
\end{enumerate}
The most interesting choices to compare are \emph{bounded-error, public-coin}
and \emph{unbounded-error, private-coin}, because, unlike the other choices,
these are not obviously weaker or stronger than the other. For a boolean matrix
$M \in \zo^{N \times N}$, we write $\Rand(M)$ for the least cost of a
\emph{bounded error, public-coin} protocol computing~$M$, and $\U(M)$ for the
least cost of an \emph{unbounded error, private-coin} protocol.

Moreover, Newman's theorem~\cite{New91} says that any bounded-error randomized protocol requires at most $O(\log\log N)$ bits of randomness. One player can privately generate these bits and send them to the other player, giving
\begin{equation}
  \U(M) \leq \Rand(M) + O(\log\log N) .
\end{equation}
Is this the best we can do in general, or can we remove the dependence on $N$?
In other words, if we fix $\Rand(M)$ to be \emph{constant} (\ie independent of the
matrix size $N$), does this imply $U(M)$ is also constant? This is equivalent to
\cref{question:intro-main}, by the following argument. Paturi and Simon
\cite{PS86} showed
\begin{equation}
  \U(M) = \log(\signrank(M))\pm 2.
\end{equation}
This means that constant $\U(M)$ is equivalent to constant sign-rank. 
Linial and Shraibman \cite{LS09} showed that margin is equivalent to discrepancy, and therefore
\begin{equation}
\Omega(\log(\margin(M)^{-1})) \leq \Rand(M) \leq O(\margin(M)^{-2}) ,
\end{equation}
so that constant $\Rand(M)$ is equivalent to constant margin. We may therefore rephrase
\cref{question:intro-main} as: \emph{If $\Rand(M) = O(1)$, is $\U(M) = O(1)$?}

The classes of communication problems with $\Rand(M) = O(1)$ have been well
studied \cite{HHH23,HWZ22,HHH22,DHPTU22,EHK22,CHZZ24,FHHH24,HR24,FGHH25,Tom25}, and several papers
\cite{HHPTZ22,HHM23,HH24,HZ24} have conjectured a negative answer to
\cref{question:intro-main}:

\begin{boxconjecture}
\label{conj:intro-main}
There exists a communication problem $M$ with $\Rand(M) = O(1)$ but $\U(M) = \omega(1)$.
\end{boxconjecture}

The most obvious candidates for this conjecture are the \textsc{$k$-Hamming
Distance} problems, which have $\Rand(\HD{k}^n) = \Theta(k \log k)$ for $k <
\sqrt n$ \cite{HSZZ06,Sag18}. Any problem which reduces to $\HD{0}$ (\ie
\textsc{Equality}) has constant sign-rank \cite{HHPTZ22}, but the question
remained open for problems which do not reduce to \textsc{Equality}, including
\textsc{$1$-Hamming Distance} and its generalizations \cite{HHH23,HWZ22, FHHH24,
FGHH25}. Several papers \cite{HHPTZ22,HHM23,HZ24} worked towards the conjecture
that for constant $k \geq 1$, these problems should satisfy
\cref{conj:intro-main}:

\begin{boxconjecture}[Now false]
\label{conj:intro-khd}
For some constant $k \geq 1$, $\signrank(\HD{k}^n) = \omega(1)$.
\end{boxconjecture}

In particular, \cite{HHPTZ22} showed that all known lower bound techniques fail
to prove this conjecture; \cite{HZ24} suggested a (now false) characterization
of the problems with both $\Rand(M) = O(1)$ and $\U(M) = O(1)$ as exactly those
which reduce to \textsc{Equality} (which for XOR functions would follow from
\cref{conj:intro-khd} for $k=1$, due to the result of \cite{CHZZ24}); and \cite{HHM23}
settled \cref{conj:intro-main} for \emph{partial} matrices by an elegant
application of the Borsuk--Ulam theorem to the \textsc{Gap Hamming Distance}
problem (where two parties given $x,y \in \zo^n$ must decide if the Hamming
distance is either at most $\alpha n$ or at least $(1-\alpha) n$ for some
constant $\alpha > 0$). Another approach to solving \cref{conj:intro-main} is to
find a completion $M$ of \textsc{Gap Hamming Distance} which
has $\Rand(M) = O(1)$; this was proven impossible in concurrent (independent)
work \cite{BHHLT25}.

\cref{thm:main} 
gives an upper bound of $\U(\HD{k}^n) = O(k)$ on the unbounded-error
communication cost, showing that it is always smaller than the bounded-error
cost $\Theta(k \log k)$, whereas \cref{conj:intro-khd} posits an arbitrarily
large gap in the other direction. However, the conjecture was sensible for
several reasons:
\begin{itemize}[itemsep=0pt]
  \item It was not known whether any other fundamentally different candidates exist; \cite{FGHH25} only recently showed that there are candidates which are irreducible to \textsc{$k$-Hamming Distance}.
  \item Earlier work on symmetric XOR problems, of which
\textsc{$k$-Hamming Distance} is the simplest example, have not witnessed any upper bounds
superior to $O(n)$ \cite{HQ17} and this is tight for \textsc{Gap Hamming Distance}~\cite{HHM23}.
  \item The Borsuk--Ulam technique can be used to give a lower bound of $\Omega(n)$ for a \emph{continuous} version of the problem, \textsc{$1$-Hamming Distance} on strings in $[0,1]^n$.
  \item Problems reducing to \textsc{Equality} 
  satisfy many nice properties that \textsc{$1$-Hamming Distance} does not, see \eg
  the close relationship between \textsc{Equality} and the $\gamma_2$-norm (see
  \cref{def:gamma-2}) \cite{HHH22,CHHS23,PSS23,CHZZ24,CHHNPS25,Tom25}, and four different
  proofs that \textsc{$1$-Hamming Distance} does not reduce to \textsc{Equality}
  \cite{HHH22,HWZ22,FHHH24,HR24}. \textsc{Equality} is
  ``special'' among the Hamming distance problems, so one may expect it to be
  special with respect to sign-rank as well.
\end{itemize}
Owing partly to the latter reason, \cite{HHPTZ22,HHM23,HZ24} stated the strongest form of
\cref{conj:intro-khd}, for $k=1$ rather than an arbitrary constant $k$. The stronger conjecture is
easier to refute than \cref{conj:intro-khd} (see \cref{rem:unit-distance}), but the easier
refutation does not generalize to $k=2$.

The main contribution of this paper is the more general technique allowing to go
beyond $k=1$, including the candidates for \cref{conj:intro-main} which
\cite{FGHH25} showed cannot be reduce to \textsc{$k$-Hamming Distance}. 

\section{Technique: Support-Rank}
\label{sec:techniques}

The main conceptual idea allowing for our upper bounds is to switch from \emph{sign-rank} to \emph{support-rank}.

\begin{boxdefinition}[Support-Rank]
Let $M \in \zo^{N \times N}$ be a boolean matrix. Its \emph{support-rank} $\suprank(M)$ is the minimal $r$ for which there exists some $A \in \bR^{N \times N}$ with rank $r$, satisfying
\[
    \forall i,j \in [N]\,\colon\qquad M(i,j) = 0 \iff A(i,j) = 0 .
\]
That is, the matrices $A$ and $M$ have the same support.
\end{boxdefinition}
Support-rank has been studied previously in the context of tensor rank~\cite{CU13,BCZ17,BCZ18}. In quantum communication complexity, it has been called \emph{nondeterministic rank}~\cite{dWol03}, in circuit complexity, \emph{equality rank}~\cite{HP2010}, and, in graph theory, \emph{minimum rank}~\cite{FH07}. It is also closely related to \emph{unit-distance graphs}~\cite{EHT65,AK14}; see our discussion in~\cref{rem:unit-distance}.

\paragraph{Why support-rank?}
A basic fact is that any boolean matrix of support-rank $r$ has sign-rank $O(r^2)$; see~\cref{sec:reductions}. The converse is false: we have $\suprank(I_N)=N$ but $\signrank(I_N)=3$ for the $N\times N$ identity matrix. Thus, proving upper bounds on support-rank is a more difficult task. Nevertheless, what is convenient about support-rank is that it behaves better than sign-rank under \emph{boolean combinations} (or \emph{reductions}). To explain this, recall that our goal is to give a sign-rank upper bound for
\[
    \HD{k}^n = \HD{\leq k}^n \wedge \neg \HD{\leq k-1}^n,
\]
where $\land$ and $\neg$ are understood entry-wise and $\HD{\leq k}^n(x,y) = 1$ if and only if $\dist(x,y) \leq k$. The challenge with sign-rank is that it is
not known whether the sign-rank of $A \wedge B$ can be bounded in terms of the sign-ranks of $A$ and~$B$~\cite{BMT21,HHPTZ22}.  So even if we can prove, say, $\signrank(\HD{\leq k}) =
O(1)$, this would not imply any bound on $\signrank(\HD{k})$. By contrast, we show the following useful properties of support-rank:
\begin{itemize}[itemsep=0pt]
	\item (\cref{sec:reductions}): We show that any matrix reducible to matrices with bounded support-rank has bounded sign-rank. All of our arguments for sign-rank upper bounds will rely on this fact.
	\item (\cref{section:veronese-maps}): We explain how to  transform polynomial identities $P(x,y) = 0$ into linear identities $\inn{x',y'} = 0$. This is useful for giving upper bounds on support-rank via polynomials.
\end{itemize}

\subsection{From Support-Rank to Sign-Rank via Reductions} \label{sec:reductions}

Reductions for constant-cost randomized communication are defined similarly to standard oracle
reductions in communication complexity (\eg \cite{BFS86,CLV19}) except that there is no bound on the
size of the oracle query inputs. For any communication problem $\cQ$ (\ie a family of boolean
matrices), and any boolean matrix $P$,  we write $\D^\cQ(P)$ for the minimum cost of a
\emph{deterministic} communication protocol computing $P$ with access to a unit-cost oracle that
computes $\cQ$. We say problem $\cP$ \emph{reduces to} $\cQ$ if there is a constant $c$ such that
for every $P \in \cP$, $\D^\cQ(P) \leq c$. More formally:

\begin{definition}[Oracle Protocols]
Let $\cQ$ be a communication problem. For any matrix $P \in \zo^{N \times N}$, we write
$\D^\cQ(P)$ for the smallest depth of a communication tree $T$, where each inner
node $v$ is labelled by a matrix $Q_v \in \cQ$ and two functions $a_v,b_v$; and
each leaf $\ell$ is labelled with an output value. On inputs $i,j \in [N]$ the
protocol at node $v$ proceeds by computing $Q_v(a_v(i),b_v(j))$ and descending
to the left or right child depending on the result. At a leaf $\ell$ the
protocol outputs the value of $\ell$, which must be equal to $P(i,j)$.
\end{definition}
Since we are concerned only with constant vs.\ non-constant costs, one may equivalently say that $
\cP$ reduces to $\cQ$ if there is a constant $c$ such that every $P \in \cP$ can be written as
\[
    P = \Gamma(Q_1, \dotsc, Q_c)
\]
for some choice of $Q_i \in \cQ$ and boolean function $\Gamma \colon \zo^c \to
\zo$ which is applied entry-wise to $Q_1, \dotsc, Q_c$ to produce $P$; see \eg
\cite{CLV19,FHHH24,FGHH25} for more details on these reductions and
\cite{ABSZ24} for applications in graph theory. It is not hard to see that
reductions preserve constant-cost randomized communication in the bounded-error
model: if $\Rand(\cQ) = O(1)$ and $\D^\cQ(\cP) = O(1)$ then $\Rand(\cP) = O(1)$,
because we may replace each query $Q \in \cQ$ with a randomized subroutine
computing $Q$ using standard majority-vote error boosting to bring the total
error down to $1/4$.

The most important property of support rank is that it allows to upper bound $\signrank(M)$ in terms
of the number of queries $\D^\cQ(M)$ required to compute $M$ with an oracle $\cQ$ that has bounded
support-rank. This lemma generalizes a theorem of \cite{HHPTZ22} which held for queries to the
\textsc{Equality} oracle.

\begin{boxlemma}
\label{lemma:supp-to-sign}
Let $\cQ$ be a family of boolean matrices with $\suprank(\cQ)\leq r$. Then, for any boolean $P$,
\[
    \signrank(P) \leq O(r^2)^{\D^\cQ(P)}.
\]
\end{boxlemma}

\begin{proof}
(Generalization of~\cite[Theorem~3.8]{HHPTZ22}.)
	Let $T$ be a decision tree for $P$ of depth $q = \D^\cQ(M)$, which queries problems in $\cQ$.  We prove that $\signrank(P) \leq \left(1+r^2\right)^q$ by induction on $q$. The base case $q = 0$ is immediate. For $q \geq 1$, let $R\in \cQ$ be the problem queried at the root of $T$. Let $A_0$ and $A_1$ be the sign matrices corresponding to the two sub-problems computed by $T$ after $R$ returns $0$ and $1$, respectively. Thus,
	\[
	P_\pm = R\circ A_1 + (\neg R)\circ A_0,
	\]	
    where $A\circ B$ is the entry-wise (or \emph{Hadamard}) product defined by $\smash{(A\circ B)_{ij} = A_{ij}B_{ij}}$.
	By the inductive hypothesis, there are real matrices $\smash{\tilde A_0}$ and $\smash{\tilde A_1}$ with rank at most $\smash{\left(1+r^2\right)^{q-1}}$ and with the same sign pattern as $A_0$ and $A_1$, respectively. Similarly, let $\smash{\tilde R}$ be a real matrix with the same support as $R$ and rank at most~$r$. Note that for a sufficiently large $\gamma > 0$, the real matrix $\tilde A_0 + \gamma \big(\tilde R\circ \tilde R\circ \tilde A_1\big)$
	has the same sign pattern as $P_\pm$. This is because on the support of $\smash{\tilde R}$, the second term will dominate, whilst the first term dictates the sign wherever $\smash{\tilde R}$ is zero. Using the fact\footnote{Write $A = \sum_{i = 1}^{r} a_iu_i^T$ and $B = \sum_{j = 1}^{s} b_jv_j^T$. Then, we have
	$A\circ B
    = \big(\sum_i a_iu_i^T\big)\circ \big(\sum_j b_iv_j^T\big)
    = \sum_{i,j} (a_iu_i^T)\circ (b_jv_j^T)
    = \sum_{i,j} (a_i\circ b_j)(u_i\circ v_j)^T$.
	This shows that $A\circ B$ can be written as the sum of $rs$ many rank-$1$ matrices.}
    that $\rank(A\circ B)\leq \rank(A)\rank(B)$, we conclude
	\[
	\signrank(P)\leq \rank\!\Big(\tilde A_0 + \gamma\big(\tilde R\circ \tilde R\circ \tilde A_1\big)\Big) \leq \rank\!\big(\tilde A_0\big) + \rank\!\big(\tilde R\big)^2\rank\!\big(\tilde A_1\big) \leq \left(1+r^2\right)^q.
    \qedhere
	\]
\end{proof}
\bigskip
\subsection{From Polynomials to Support-Rank via Veronese Maps}
\label{section:veronese-maps}

Note that the support-rank of $M\in\{0,1\}^{N\times N}$ is at most $r$ if and only if there are vectors $u_1,u_2,\dots, u_N\in\R^r$ and $v_1,v_2,\dots, v_N\in\R^r$ such that
\begin{equation}
\label{eq:example-veronese-lin}
	\forall i,j\in[n]\,\colon\qquad M_{ij} = 0 \iff \langle u_i, v_j \rangle = 0.
\end{equation}
It is more convenient to work with polynomial equations $P(u_i, v_j) = 0$. For example, suppose we have
\begin{equation}
\label{eq:example-veronese-poly}
	\forall i,j\in[n]\,\colon\qquad M_{ij} = 0 \iff P(u_i, v_j) = 0
\end{equation}
where, say, $u_i, v_j \in \bR^2$ and $P(a, b)$ is a polynomial on 4 variables, say $P(a,b) = a_1^2 +
b_1^2 - a_1b_1 + 3a_2b_2$.  Then we can write $P$ as an inner product of two vectors,
each depending only on $a$ or only on $b$, by grouping each monomial into its
own dimension:
\[
  P(a, b) = \left\langle (a_1^2, 1, -a_1, 3a_2), \; (1, b_1^2, b_1, b_2) \right\rangle .
\]
In this way we transform equations like \cref{eq:example-veronese-poly} into the equations like
\cref{eq:example-veronese-lin} required for the definition of support-rank, where the dimension
is at most the number of monomials in $P$. In general we have the following proposition (whose proof simply generalizes the above discussion and is hence omitted).

\begin{boxproposition}[Veronese Map]\label{lemma:veronese}
	Let $M\in\{0,1\}^{n\times n}$ and let $P$ be a real polynomial in $2m$ variables. Assume that there are functions $\alpha_t\colon [n]\to \R$ and $\beta_t\colon [n]\to \R$ for $t\in [m]$ that satisfy
\[
	\forall i,j\in[n]\,\colon\qquad M_{ij} = 0 \iff P\big(\alpha_1(i),\ \dots,\ \alpha_m(i),\ \beta_1(j),\ \dots,\ \beta_m(j)\big) = 0.
\]
	Then, $\suprank(M)$ is at most the number of monomials in $P$ with non-zero coefficients.
\end{boxproposition}

\subsection{Support-Rank and Unit-Distance Graphs}
\label{rem:unit-distance}
A \emph{(faithful) unit-distance graph}~\cite{EHT65,AK14} in dimension $d$ is a graph $G = (V,E)$ whose vertices $x \in V$ can be identified with points $u_x \in \bR^d$ such that
\[
\{x,y\} \in E \iff 
\|u_x - u_y\|_2 = 1.
\]
We claim that the complement of the adjacency matrix of a unit-distance graph in dimension $d$ has support-rank at most $O(d)$. Indeed, we have $\|u_x - u_y\|_2 = 1 \Leftrightarrow P(u_x,u_y)=0$ for the $2d$-variate polynomial $\smash{P(a,b)\coloneqq \sum_{i=1}^d (a_i-b_i)^2 - 1}$ with $O(d)$ monomials, and the claim follows from~\cref{lemma:veronese}. Conversely, it is easy to show that that any boolean matrix $M$ with $\suprank(M) = d$ is the complement
of the bi-adjacency matrix of a bipartite unit-distance graph in dimension~$d$ (indeed, normalize all $u_i$, $v_j$ to have length $1/\sqrt{2}$).

It is a classic fact~\cite{EHT65} that the hypercube graph (the bipartite graph with
bi-adjacency matrix~$\HD{1}^n$) is a unit-distance graph in dimension $2$. By the above discussion, we have $\suprank(\neg\HD{1}^n)\leq O(1)$, which implies $\signrank(\HD{1}^n)\leq O(1)$ via \cref{lemma:supp-to-sign}. This already proves \cref{thm:main} in the special case~$k=1$. However, this argument does not generalize to $k = 2$, because $\neg\HD{2}^n$ (and also $\HD{2}^n$) contains an unbounded-size identity submatrix, which shows that its support-rank is unbounded.

\section{Sign-Rank of \texorpdfstring{$k$}{k}-Hamming Distance} \label{sec:proof-of-main-thm}

In this section, we prove \cref{thm:main}. We have $\HD{k}^n = \HD{\geq k}^n\land \neg\HD{\geq k+1}^n$ and hence by \cref{lemma:supp-to-sign} it suffices to prove an upper bound on the support-rank of $\HD{\geq k}^n$.
\begin{boxtheorem}\label{thm:HD-suprank}
We have $2^k \leq \suprank(\HD{\geq k}^n) \leq 4^k$ for all $k\leq n$.
\end{boxtheorem}
It is important here to consider $\HD{\geq k}^n$ rather than $\HD{\leq k}^n$, as the latter contains an identity submatrix
of size $\Omega(2^n)$ and therefore its support-rank depends on $n$. We also note that this theorem, together with \cref{lemma:supp-to-sign}, implies more generally that any matrix reducible to $\HD{k}^n$ has constant sign-rank.
\begin{corollary}\label{cor:reductions-to-HD}
    For all boolean matrices $M$, $\signrank(M) = 2^{O(k \cdot \D^{\HD{k}}(M))}$ and $\U(M) = O(k \cdot \D^{\HD{k}}(M) )$.
\end{corollary}

The lower bound in \cref{thm:HD-suprank} follows directly from the fact that $\HD{\geq k}$ contains an identity submatrix of size $2^k$, induced by the set of all strings ending in $n-k$ zeros. 

The upper bound uses the following method (which we further generalize in \cref{section:rank-problems}). We show that there exists a map $A \colon \zo^n \to \bR^{k \times k}$ assigning to each binary string $x\in \zo^n$ a $k\times k$ matrix $A(x)$ with the property that
\begin{equation}
\label{eq:intro-hd-rankproblem}
    \forall x,y \in \zo^n \,\colon\qquad \HD{\geq k}^n(x,y) = 1 \iff \rank\!\big( A(x) - A(y) \big) = k .
\end{equation}
In other words, the output of the communication problem depends only on whether the matrix $A(x)-A(y)$ has full rank. This can be verified by testing if its determinant is $0$. Since the determinant is given by a polynomial in the entries of the matrix, we can then use a Veronese map to obtain a support-rank upper bound. \cref{fig:khd-proof} illustrates the proof that follows.

\begin{figure}[H]
	\begin{center}
		\input{pics/khd.tex}
	\end{center}
        {\color{captioncolor}
	\[
	\dist(x,y) \geq k \Leftrightarrow \rank\!\big(\Diag(x-y)\big) \geq k\Leftrightarrow  \rank\!\big(\Pi\big(\Diag(x-y)\big)\big) = k \Leftrightarrow \det\!\big(\Pi\big(\Diag(x-y)\big)\big) \neq 0
	\]\vspace{-5mm}
        }
        \caption{A sketch of the argument: We reduce $\HD{\geq k}$ to checking whether a polynomial vanishes.}
        \label{fig:khd-proof}
\end{figure}

As a first step, observe that if we view $\zo^n$ as a subset of $\R^n$, then
\begin{equation}\label{eq:first-step}
\forall x,y\in\{0,1\}^n\,\colon\quad \HD{\geq k}^n(x,y) = 1 \iff x-y \text{ has $\geq k$ non-zero entries} \iff \rank(\Diag(x-y)) \geq k.
\end{equation}
The matrices $\Diag(x-y)$ are of size $n \times n$ and the next step is to reduce their size to $k\times k$, without changing the rank, provided it is at most $k$.
Informally, this can be accomplished by applying a random projection which has
probability $1$ of preserving ranks. Formally, we have the following lemma.
\begin{boxlemma}[Rank Compression]
    \label{lemma:rank-compression}
    Let $\cM$ be a finite set of matrices in $\R^{a\times b}$. For any $a'$ and $b'$, there exists a linear map $\Pi\colon \R^{a\times b}\to\R^{a'\times b'}$ which satisfies
    \[
    \forall M\in \cM\,\colon\qquad \rank\big(\Pi(M)\big) = \min\!\big(\rank(M),\ a',\ b'\big).
    \]
\end{boxlemma}
\begin{proof}
It suffices to prove the statement with $b = b'$, since the general result is
recovered by applying this case twice, transposing in between. We can further
assume $a' < a$, as otherwise we can just take any injective linear map for
$\Pi$. For each $M\in\cM$, pick a subspace $V_M\subseteq \img(M)\subseteq \R^a$
of dimension $\min(\rank(M), a')$. Since $\cM$ is finite, there is a subspace
$V\subseteq\R^a$ of dimension $a-a'$ such that
\[
V\cap \bigcup_{M\in\cM} V_M = \{0\}.
\]
We now let $\Pi(x) = Px$ where $P\in\R^{a'\times a}$ is the projection with kernel $V$. Thus, for all $M\in\cM$ we have
\[
\rank\big(\Pi(M)\big) = \rank(PM) = a' - \dim\!\big(V \cap \img(M)\big) \geq \dim(V_M) = \min\!\big(\rank(M),\ a'\big).
\]
The converse inequality clearly holds as well.
\end{proof}
In particular, using $\cM = \{\Diag(x-y)\mid x,y\in \zo^n\}$ and the fact that
$k\leq n$, we obtain a linear map~$\Pi\colon\R^{n\times n}\to
\R^{k\times k}$ with the desired property
\begin{equation}\label{eq:khd-rank-problem}
\forall x,y\in \{0,1\}^n\,\colon\qquad \rank\!\Big(\Pi\big(\Diag(x-y)\big)\Big) = k \iff \rank\!\big(\Diag(x-y)\big) \geq k
\iff \dist(x,y) \geq k.
\end{equation}
Thus, setting $A(x) = \Pi(\Diag(x))$ for all $x\in\zo^n$, we obtain the desired characterisation of \Cref{eq:intro-hd-rankproblem}.

It remains to express the rank bound of this $k\times k$ matrix in terms of a polynomial suitable for the Veronese map. Writing $S_k$ for the set of permutations on $[k]$, we further have
\[
\forall x,y\in \{0,1\}^n\,\colon\qquad \det\!\big(A(x)-A(y)\big)\big) = \sum_{\pi\in S_k}\sign(\pi)\prod_{i = 1}^k \Big(A(x)_{i\pi(i)}-A(y)_{i\pi(i)}\Big).
\]
Expanding the right side of the above, we obtain a polynomial over $2k$ variables, corresponding to the entries of $A(x)$ and $A(y)$, with at most $2^k\cdot k!$ many monomials. Thus, using the Veronese map from \Cref{lemma:veronese}, we conclude that
\[
\suprank(\HD{\geq k}^n) \leq k!\cdot 2^k.
\]
This bound is already independent of $n$, but for \cref{thm:HD-suprank} we claim an even better upper bound of $4^k$. To achieve this, we can be smarter when constructing the Veronese map. In the following section, we prove the following fact, which gives an improved Veronese map, concluding the proof of \cref{thm:HD-suprank}.

\begin{proposition}\label{fact:det}
    For matrices $A,B\in \R^{k\times k}$, we can write $\det(A-B)$ as a sum of at most $4^k$ many terms, each of the form $\pm \det(A')\det(B')$ where $A'$ and $B'$ denote some square submatrices of $A$ and $B$, respectively.
\end{proposition}

\begin{remark}
\label{remark:large-alphabet}
The above proof works for \textsc{$k$-Hamming Distance} over any finite alphabet, not only for~$\zo$. One simply has to identify the alphabet with some subset of $\R$ to satisfy \cref{eq:first-step}.
\end{remark}

\subsection{Proof of \texorpdfstring{\cref{fact:det}}{Proposition}}\label{sec:improving-bound}

We use the following lemma from \cite{Mar90}, whose proof we include for completeness. Note that, in this lemma, the right side of the equation contains $\sum_{k = 0}^n\binom{n}{k}^2 = \binom{2n}{n}\leq 4^n$ many terms, which proves \cref{fact:det}.

\begin{boxlemma}\label{lemma:det-of-sum}
	Let $n\in\bN$ and let $A,B\in \R^{n\times n}$. Then,
	\[
	\det(A+B) = \sum_{\substack{\alpha,\beta \subseteq [n]\\[1mm] |\alpha| = |\beta|}} (-1)^{s(\alpha)+s(\beta)}\det(A_{|\alpha\times \beta})\det(B_{|\bar\alpha\times \bar\beta}),
	\]
	where $s(\cdot)$ denotes the sum of a set, and $A_{|\alpha \times \beta}$ denotes the submatrix of $A$ indexed by sets $\alpha$ and $\beta$.
\end{boxlemma}

\begin{figure}[t]
	\begin{center}
		\input{pics/detsum.tex}
	\end{center}
	\vspace{-2mm}
	{\color{captioncolor}
		\[
		B_{15}A_{23}A_{32}A_{44}B_{56}B_{61} + B_{15}A_{24}A_{33}A_{42}B_{56}B_{61} = (A_{23}A_{32}A_{44}+A_{24}A_{33}A_{42})(B_{15}B_{56}B_{61})
		\]}
	\vspace{-5mm}
	\caption{Two terms in the expansion of $\det(A+B)$ that only differ on the square given by $\alpha\times\beta$ and can thus be factored. Doing this for all like terms of each square, we arrive at \Cref{lemma:det-of-sum}.}
	\label{fig:det-sum}
\end{figure}

\begin{proof}
Let $S_n$ denote the set of all permutations on the set $[n]$. We have
	\begin{align*}
		\det(A+B) &= \sum_{\pi\in S_n}\sign(\pi)\prod_{i = 1}^n(A+B)_{i\pi(i)}
		= \sum_{\alpha\subseteq [n]}\ \sum_{\pi\in S_n}\sign(\pi)\ \prod_{i \in \alpha} A_{i\pi(i)}\prod_{i \in \bar\alpha} B_{i\pi(i)}.
	\end{align*}
	Given a choice of $\alpha \subseteq [n]$, any permutation $\pi\in S_n$ can be uniquely decomposed as $\pi = \tau\circ\pi_{\alpha}\circ \pi_{\bar\alpha}$ where
	\[
	\pi_{\alpha}(x) = x \text{ for all } x\in\bar\alpha\qquad\text{and}\qquad \pi_{\bar\alpha}(x) = x \text{ for all } x\in\alpha,
	\]
	and $\tau\in S_n$ is the unique permutation with $\tau(\alpha) = \pi(\alpha)$ that is order preserving on both $\alpha$ and $\bar\alpha$. Note that $\tau$ only depends on $\pi(\alpha)$, which we now call $\beta$. Moreover, after fixing $\beta$, the correspondence between pairs $(\pi_\alpha,\pi_{\bar\alpha})$ and $\pi$ is bijective. Thus,
\allowdisplaybreaks
\begin{align*}
		\det(A+B) &= \sum_{\substack{\alpha,\beta\subseteq [n]\\[1mm]|\alpha|=|\beta|}}\sum_{\pi_\alpha}\sum_{\pi_{\bar\alpha}}\sign(\tau\circ \pi_{\alpha}\circ \pi_{\bar\alpha})\ \prod_{i \in \alpha} A_{i(\tau\circ \pi_\alpha)(i)}\prod_{i \in \bar\alpha} B_{i(\tau\circ \pi_{\bar\alpha})(i)}\\
		&= \sum_{\substack{\alpha,\beta\subseteq [n]\\[1mm]|\alpha|=|\beta|}}\sign(\tau)\left[\sum_{\pi_\alpha}\sign(\pi_{\alpha})\prod_{i \in \alpha} A_{i(\tau\circ \pi_\alpha)(i)}\right]\left[\sum_{\pi_{\bar\alpha}} \sign(\pi_{\bar\alpha})\prod_{i \in \bar\alpha}B_{i(\tau\circ \pi_{\bar\alpha})(i)}\right]\\[2mm]
		&= \sum_{\substack{\alpha,\beta\subseteq [n]\\[1mm]|\alpha|=|\beta|}}\sign(\tau)\det(A_{|\alpha\times \beta})\det(B_{|\bar\alpha\times \bar\beta}).
	\end{align*}
	Finally, note that the number of inversions of $\tau$ is given by
	$\sum_{i\in\alpha}|\pi(i)-i|$, which has the same parity as $\sum_{i\in\alpha}(\pi(i)+i) = s(\beta) + s(\alpha)$.
	Thus, $\sign(\tau) = (-1)^{s(\alpha)+s(\beta)}$, just as desired.
\end{proof}

\section{Rank Problems and Generalizations of \texorpdfstring{$k$}{k}-Hamming Distance}
\label{section:rank-problems}

\cref{cor:reductions-to-HD} showed that every problem that reduces in $q$ queries to \textsc{$k$-Hamming Distance} has sign-rank~$2^{O(qk)}$. But \cite{FGHH25} recently showed that there
exist problems with $\Rand(M) = O(1)$ that \emph{do not} reduce to
\textsc{$k$-Hamming Distance}. At first, this provides some hope of using these
problems to separate constant margin from constant sign-rank, but we shall crush
this hope, using another new idea.

The problems of \cite{FGHH25} were constructed via \emph{distance-$r$
compositions}. The simplest example, which was the main focus of their paper, is
the \textsc{$\{4,4\}$-Hamming Distance} problem. We will discuss only this
example here and leave the formal definition of \emph{distance-$r$ compositions} for
\cref{section:rank-closure-composition}.

\begin{boxexample}[$\{4,4\}$-Hamming Distance]
\label{ex:44-hamming-distance}
Alice and Bob receive matrices $X, Y \in \zo^{n \times n}$ respectively. Write
$X_i$ for the $i^\mathrm{th}$ row of $X$, and similar for $Y$. The players
should output 1 if and only if the following conditions are satisfied:
\begin{enumerate}
    \item There are at most 2 rows $i,j \in [n]$ such that $X_i \neq Y_i$ and $X_j \neq Y_j$; and
    \item For each row $i$ where $X_i \neq Y_i$, it holds that $\dist(X_i, Y_i) \leq 4$.
\end{enumerate}
\end{boxexample}
In essence, the reason that \textsc{$\{4,4\}$-Hamming Distance} does not reduce
to \textsc{$k$-Hamming Distance} for any constant $k$ is that a
\textsc{$k$-Hamming Distance} query is not capable of distinguishing between two
rows each of distance 4 (where the correct output should be 1), and two rows of
distance 6 and 2 (where the correct output should be 0). This poses a challenge
for sign-rank as well, since a na\"ive application of our method for
\textsc{$k$-Hamming Distance} encounters the same issue.

To handle this type of problem, we define a class of
problems called \emph{rank problems}. A rank problem is any problem which can be
expressed as a function of the rank of the sum of matrices held by Alice and Bob; in other words, any problem that can
put in the form similar to \cref{eq:intro-hd-rankproblem} in our upper bound of
\textsc{$k$-Hamming Distance}. This notion of a rank problem resembles the problems defined in \cite{SS24}, but concerning matrices over $\R$ instead of a finite field $\bF$. This difference is crucial in our argument.

\begin{boxdefinition}[Rank Problem] \label{def:rank-problem}
A boolean matrix $P \in \zo^{N \times N}$ is a \emph{rank problem} of
\emph{order} $k$ if for some $a,b$ and every $x \in [N]$, there exist real matrices $A(x)$ and $B(x)$ in $\R^{a\times b}$, satisfying
\[
	\forall x,y\in[N]\,\colon\qquad P(x,y) = g\Big(\rank\!\big(A(x) + B(y)\big)\Big),
\]
where $g \colon \{0,1,2,\dots\} \to
\zo$ is a function which is constant for inputs $\geq k$. 
We say that a family $\cP$ of boolean matrices is \emph{family of rank problems} of
order $k$ if there is some function $g$ such that
each $P \in \cP$ is a rank problem of order $k$ with associated function $g$. Moreover, we call the rank problem \emph{symmetric} if $A(x) + B(x) = 0$ for all $x \in [N]$.
\end{boxdefinition}

\begin{remark}\label{rem:rank-domain}
	In the above definition, we can equivalently take $a = b = k$. This is because we can always embed $A$ and $B$ in larger matrices, or compress them using \Cref{lemma:rank-compression} without changing the ranks pertinently. We also note that it suffices to define $g$ on the domain $\{0,1,\dots, k\}$.
\end{remark}

\begin{example}\label{ex:HD-rank-problem}
In the proof of \cref{thm:HD-suprank}, we showed in \cref{eq:khd-rank-problem} that
$\HD{\geq k}$ is a symmetric rank problem of order $k$, with associated function $g(t) =
\ind{t \geq k}$.
\end{example}

We prove the following three properties of rank problems.

\begin{boxtheorem}\label{thm:rank-problem-main}
	Rank problems satisfy the following three properties:
	\begin{enumerate}
		\item Rank problems of constant order have constant sign rank.
		\item Rank problems of constant order are closed under reductions.
		\item Symmetric rank problems of constant order are closed under distance-$r$ compositions.
	\end{enumerate}
\end{boxtheorem}

Since \textsc{$k$-Hamming Distance} is a rank problem of constant order, the
above 3 properties guarantee that any problem obtained by reductions and
compositions of it will have constant sign-rank. As explained in \cite{FGHH25},
this covers all known examples of problems with $\Rand(M) = O(1)$, apart from the remaining separation candidates from \cite{SS24}.

In the following three subsections, we will give proofs of each of the three parts in the above
theorem, with quantitative bounds.

\subsection{Rank Problems have Bounded Sign-Rank}

\begin{boxlemma}[\Cref{thm:rank-problem-main}, Part 1]\label{lemma:signrank-rank}
	If $P$ is a rank problem of order $k$, then
	\[
	\signrank(P) = 2^{O(k\log k)}.
	\]
\end{boxlemma}

We first deal with a special class of simple rank problems, and then reduce general rank problems to this case. Let $g \colon \{0, 1, \dotsc\} \to \zo$ be some function as in the definition of rank problems. We write $\partial g$ for the number of times $g$ changes value. If $\partial g \leq 1$, then we say that the associated rank problem is \emph{monotone}.

\begin{lemma}\label{lemma:suprank-monotone}
	If $P$ is a monotone rank problem of order $k$, then 
	\[
	\min\!\Big(\suprank(P),\ \suprank(\neg P)\Big) \leq 4^k.
	\]
\end{lemma}

\begin{proof}
	Let $P\in\{0,1\}^{N\times N}$ be a monotone rank problem of order $k$. Without loss of generality, replacing $P$ with $\neg P$ if necessary, there are functions $A,B\colon [N] \to \R^{a\times b}$ and some $r\leq \min(k,a,b)$ that satisfy
	\[
	\forall x,y\in[N]\,\colon\qquad P(x,y) = 1 \iff \rank\!\big(A(x)+B(y)\big) \geq r.
	\]
	But now, we finish just like in the proof of \Cref{thm:HD-suprank} using rank compression and the Veronese map: By \Cref{lemma:rank-compression}, there is a rank-preserving linear map $\Pi:\R^{a\times b}\to \R^{r\times r}$ such that
	\begin{align*}
		\forall x,y\in[N]\,\colon\qquad P(x,y) = 1 &\iff \rank\!\big(A(x)+B(y)\big) \geq r\\
		&\iff \rank\!\Big(\Pi(A(x))+\Pi(B(y))\Big) = r\\
		&\iff \det\!\Big(\Pi(A(x))+\Pi(B(y))\Big) \neq 0.
	\end{align*}
	The determinant can be expanded using \Cref{lemma:det-of-sum} to obtain a polynomial. Using the Veronese map from \Cref{lemma:veronese} together with \Cref{fact:det}, we conclude that $\suprank(P)\leq 4^r\leq 4^k$.
\end{proof}

We now show that any rank problem reduces to a bounded number of monotone rank problems.

\begin{lemma}
If $P$ is a rank problem of order $k$, then there exists a family $\cQ$
of monotone rank problems of order $k$, such that
\[
    \D^Q(P) = O(\log k).
\]
\end{lemma}

\begin{proof}
	For $P\in\{0,1\}^{N\times N}$, let $A,B\colon [N] \to \R^{a\times b}$, and $g\colon \{0,1,2,\dots\}\to \{0,1\}$ constant on $[k,\infty )$, such that
\[
	\forall x,y\in[N]\,\colon\qquad P(x,y) = g\Big(\rank\!\big(A(x)+B(y)\big)\Big).
\]
	Using binary search, $\rank(A(x)+B(y))$ can be determined exactly by a decision tree of depth $O(\log k)$ that makes queries to the family of monotone rank problems of order $k$, given by the same maps $A$ and $B$. Since $P$ only depends on this rank, we conclude that $\D^Q(P) = O(\log k)$.
	
	We note that this bound can be improved to $\D^Q(P) = O(\log \partial g)$, since we do not need to determine $\rank(A(x)+B(y))$ exactly, but only need to determine whether or not it lies in the support of $g$.
\end{proof}

The proof of \Cref{lemma:signrank-rank} is now a simple application of \Cref{lemma:suprank-monotone} and \Cref{lemma:supp-to-sign}.

\subsection{Rank Problems are Closed under Reductions}
\label{section:rank-closure-reductions}

\begin{boxlemma}[\Cref{thm:rank-problem-main}, Part 2]
\label{thm:rank-problem-reduction}
If the problem $P$ is a boolean combination of $q$ rank problems of order $k$
each, then $P$ is a rank problem of order $O(k)^q$.
\end{boxlemma}

\begin{proof}
	Let $P\in \{0,1\}^{N\times N}$ be a problem that is the combination of rank problem instances $Q_1,Q_2,\dots, Q_q$ of order $k$ each. For each $i\in [q]$, if $Q_i\in \{0,1\}^{N_i\times N_i}$, then by \Cref{rem:rank-domain}, there are $A_i,B_i\colon [N_i]\to
    \R^{k\times k}$ and $g_i\colon \{0,1,\dots, k\}\to \{0,1\}$, such that 
\[
	\forall x,y\in[N_i]\,\colon\qquad  Q_i(x,y) = g_i\Big(\rank\!\big(A_i(x)+ B_i(y)\big)\Big).
\]
For each $x\in [N]$, we now define the diagonal block matrix $A(x)$ with $q$
blocks as follows. For each $i\in [q]$, let $x_i$ denote the input to $Q_i$
corresponding to input $x$ of $P$. We put $(k+1)^{i-1}$ many copies of
$A_i(x_i)$ on the diagonal, as illustrated in \cref{fig:rank-reduction}:
	
	\begin{figure}[H]
		\begin{center}
			\input{pics/rank-reduction.tex}
		\end{center}
		\caption{Construction of the matrix $A(x)$ from the matrices $A_i(x_i)$ with $i\in [q]$.}
        \label{fig:rank-reduction}
	\end{figure}
\noindent
We similarly define the diagonal block matrix $B(x)$. Now, we have
\[
    \forall x,y\in[N]\,\colon\qquad \rank\!\big(A(x)+B(y)\big) = \sum_{i = 1}^q (k+1)^{i-1}\rank\!\Big(A_i(x_i)+B_i(y_i)\Big).
\]
	Since $\rank(A_i(x)+B_i(y))\leq k$ for all $i\in[q]$,
\[
	\forall x,y\in[N]\,\colon\qquad  Q_i(x,y) = g_i\Big(\rank\!\big(A_i(x_i)+B_i(y_i)\big)\Big) = g_i\left(\left\lfloor\frac{\rank\!\big(A(x)+B(y)\big)}{(k+1)^{i-1}}\right\rfloor \text{ modulo }k+1\right).
\]
	Thus, $\rank(A(x)+B(y))$ encodes the answer to every query required to
	compute $P(x,y)$ as the digits in its base $k+1$ expansion. We conclude that there is some function $g:\{0,1,\dots\}\to\{0,1\}$ that is constant for inputs at least $(k+1)^q-1$ and such that
	\[
	\forall x,y\in[N]\,\colon\qquad P(x,y) = g\Big(\rank\!\big(A(x)+B(y)\big)\Big).
	\]
	Thus, $P$ is a rank problem of order $(k+1)^q-1 = O(k)^q$.
\end{proof}

\begin{remark} \label{rem:symmetry}
	Since the construction in the above proof preserves symmetry, we also find that symmetric rank problems of constant order are closed under reductions.
\end{remark}

\begin{example}
	Let $S\subseteq \bN$ be a finite set. Consider the problem $P\in\{0,1\}^{2^n\times 2^n}$ of deciding whether two strings $x,y\in \{0,1\}^n$ satisfy $\dist(x,y)\in S$. Note that $P$ can be computed by $O(|S|)$ many queries to $\HD{\geq k}$ with $k\leq 1+\max S$. By \Cref{ex:HD-rank-problem}, all of these problems are rank problems of order $O(\max S)$. Thus, by \Cref{thm:rank-problem-reduction}, it follows that $P$ is a rank problem of order $O(\max S)^{O(|S|)}$.
\end{example}

\subsection{Rank Problems are Closed under Distance-\texorpdfstring{$r$}{r} Compositions}
\label{section:rank-closure-composition}

We now define the distance-$r$ compositions of \cite{FGHH25}. For simplicity, we
present a definition in which each inner problem $P_i$ is boolean-valued,
although \cite{FGHH25} allows an arbitrary constant-size range of values. It is
not difficult to show that the latter can be reduced to queries of the former.

\newcommand{\ldsbracket}{[\![}
\newcommand{\rdsbracket}{]\!]}
\begin{boxdefinition}
\label{def:distance-r-composition}
Fix any $r$ and an \emph{outer} function $h \colon \{0, \dotsc, r\} \to
\zo$. For boolean matrices $P_1, \dotsc, P_m \in \zo^{N \times N}$, we define their
\emph{distance-$r$ composition}
$h \ldsbracket P_1, \dotsc, P_m \rdsbracket \colon [N]^m \times [N]^m \to \zo$
as follows. For any $x,y \in [N]^m$, write $\Delta(x,y) \define \{ i \in [m] \;|\; x_i \neq y_i \}$. Then,
\[
    h\ldsbracket P_1, \dotsc, P_m \rdsbracket (x,y)
    = \begin{cases}
        0 &\text{ if } |\Delta(x,y)| > r, \\
        h\left(\sum_{i \in \Delta(x,y)} P_i(x_i, y_i) \right) &\text{ otherwise.}
    \end{cases}
\]
\end{boxdefinition}

\begin{example}
    The problem $\HD{k}^n$ is the distance-$k$ composition where we take $h(t) =\ind{t=k}$ and take each $P_i$ to be $\neg I_{2,2}$, the negation of the $2 \times 2$ identity matrix.
\end{example}

\begin{example}
    The \textsc{$\{4,4\}$-Hamming Distance} problem (\cref{ex:44-hamming-distance})
    is the distance-$2$ composition where we take $h(t) = \ind{t \leq 2}$
    and take each $P_i$ to be $\HD{\leq 4}^n$.
\end{example}

\begin{boxlemma}[\Cref{thm:rank-problem-main}, Part 3]
\label{thm:rank-problem-composition}
Let $\cP$ be a family of symmetric rank problems of order $k$ and let $P_1, \dotsc, P_n \in \cP$. Then, for any $h \colon \{0, \dotsc, r\} \to \zo$, the distance-$r$ composition
\[
	P = h\ldsbracket P_1, \dotsc, P_n \rdsbracket
\]
is a symmetric rank problem of order $O(rk)^{O(rk^2)}$.
\end{boxlemma}

\begin{proof}
We may assume without loss of generality that each $P_i$ is an $N \times N$
matrix. The goal is to write $P$ as a boolean combination of $O(rk^2)$ many rank
problems, each of order $O(rk)$; from there, the conclusion holds by
\cref{thm:rank-problem-reduction}. First consider the case where $h =
\mathbbm{1}$ is the constant $1$ function. In this case
\[
  \forall x,y\in[N]^n \;\colon\qquad
    P(x,y) = \mathbbm{1} \ldsbracket P_1, \dotsc, P_n \rdsbracket (x,y) = 1 \iff |\{i\in [n] \mid x_i\neq y_i\}| \leq r .
\]
This is just the Hamming distance problem $\neg \HD{\geq r+1}^n$ over the  alphabet
$\Sigma = \{0,1\}^N$, which is a rank problem of order $k+1$ (see
\cref{remark:large-alphabet}). We now prove the statement for general $h$. Let $g \colon
\{0,1,2,\dotsc\}\to \{0,1\}$ and $A_i(x) \in\R^{a_i\times b_i}$ for $i\in [M]$
be such that $g$ is constant on inputs at least $k$, and
\[
	\forall x_i, y_i \in[N] \;\colon\qquad P_i(x_i, y_i) = g\Big(\rank\!\big(A_i(x_i)-A_i(y_i)\big)\Big).
\]
We wish to recover the multiset $\big\{ \rank\!\big(A_i(x_i) - A_i(y_i)\big) \;\mid\; i \in [n],\ 
	x_i \neq y_i \big\}$.
Due to the simple fact stated below, it suffices to determine the values
\begin{equation}
\label{eq:capped-multiset}
    \forall t \in [k] \;\colon\qquad \sum_{i=1}^n \min\Big( \rank\!\big(A_i(x_i) - A_i(y_i)\big),\ t \Big).
\end{equation}

\begin{fact}\label{fact:multiset}
Let $U$ and $V$ be two multisets whose elements are integers from $\{0, \dotsc, s\}$. If for all $t \in [s]$ we have $\sum_{u \in U} \min(u, t) = \sum_{v \in V} \min(v, t)$, then necessarily $U = V$.
\end{fact}
Now fix any $t \in [k]$. We determine the value of \cref{eq:capped-multiset} by constructing a diagonal block matrix whose blocks are compressions of the matrices $A_i(x_i)$ for $i\in [n]$, capping their ranks at $t$, as follows. 

For every $i\in[n]$, by \Cref{lemma:rank-compression}, there is a linear map $\Pi^{(t)}_i \colon \R^{a_i\times b_i}\to \R^{t \times t}$ such
that
\[
    \forall x_i, y_i \in[N] \;\colon\qquad
        \rank\!\Big(\Pi^{(t)}_i\big(A_i(x_i)-A_i(y_i)\big)\Big)
        = \min\!\Big(\rank\!\big(A_i(x_i)-A_i(y_i)\big),\ t\Big).
\]
For every $x \in [N]^n$ and all $i\in [n]$, define $B^{(t)}_i(x_i) =
\Pi^{(t)}_i(A_i(x_i)) \in \R^{t\times t}$ and construct the diagonal block matrix $B^{(t)}(x)\in \R^{nt \times nt}$ whose blocks are given by $B^{(t)}_i(x_i)$, illustrated in \Cref{fig:composition}.
\begin{figure}[t]
	\begin{center}
		\input{pics/rank-composition.tex}
		\caption{Construction of $B^{(t)}(x)$ from the $B^{(t)}_i(x_i)$ for $i\in [n]$.}
		\label{fig:composition}
	\end{center}
\end{figure}
By \Cref{lemma:rank-compression}, there is a linear map $\Pi^{(t)}\colon \R^{nt\times
nt}\to \R^{rt\times rt}$ that satisfies
\[
	\forall x,y\in[N]^n\,\colon\qquad \rank\!\Big(\Pi^{(t)}\big(B^{(t)}(x)-B^{(t)}(y)\big)\Big)
            = \min\!\Big(\rank\!\big(B^{(t)}(x)-B^{(t)}(y)\big),\ rt \Big).
\]
Now set $A^{(t)}(x) = \Pi^{(t)}(B^{(t)}(x))$ for all $x\in [N]^n$. We claim that $P(x,y)$ is a function of
\[
	\mathbbm{1}\ldsbracket P_1, \dotsc, P_n \rdsbracket (x,y)
            \quad\text{and}\quad
        \rank\big(A^{(t)}(x)-A^{(t)}(y)\big) \quad \text{for} \quad 1\leq t\leq k.
\]
If $\mathbbm{1}\ldsbracket P_1, \dotsc, P_n \rdsbracket (x,y) = 0$, then $P(x,y)
= 0$. Else, we must have $|\{i\in [M]\mid x_i\neq y_i\}| \leq r$ which implies
that for all $t \in [k]$, $\rank(B^{(t)}(x)-B^{(t)}(y)) \leq rt$ and thus
\[
 \forall x,y\in[N]^n \;\colon\quad \rank\!\Big(A^{(t)}(x)-A^{(t)}(y)\Big)
    = \rank\!\Big(B^{(t)}(x)-B^{(t)}(y)\Big)
    = \sum_{i = 1}^n \min\!\Big(\rank\!\big(A_i(x)-A_i(y)\big),\ t\Big).
\]

Now for fixed $x,y\in [N]^m$, the above expression is a function of all the $O(rk)$ many possible monotone rank problems that correspond to the map $A^{(t)}$, each of order $rt$.
Doing this for all $t\in [k]$ and using \Cref{fact:multiset}, we deduce that the multiset $\{\rank(A_i(x)-A_i(y))\mid i\in [M]\}$ is a boolean combination of $O(rk^2)$ many rank problems of order $rk$ each. But then the same is true for the value of
\[
	P(x,y) = h\!\left(\sum_{i \in \Delta(x,y)}P_i(x_i,y_i)\right) = h\!\left(\sum_{i \in \Delta(x,y)} g\Big(\rank\!\big(A_i(x)-A_i(y)\big)\Big)\right).
\]
Thus, we have expressed $P$ as a boolean combination of the rank problem $\mathbbm{1}\ldsbracket P_1, \dotsc, P_n \rdsbracket$ of order $k+1$, together with $O(rk^2)$ many symmetric rank problems of order $rk$ each. By \Cref{thm:rank-problem-reduction} and \Cref{rem:symmetry}, we conclude that $P$ is a symmetric rank problem of order $O(rk)^{O(rk^2)}$.
\end{proof}

\section{What's SUPP? Complexity Classes and their Relations}
\label{section:whatsupp}

\subsection{New Classes}
One of the main ideas that allowed us to get general upper bounds on sign-rank
was the use of support-rank as an intermediate step, and we have shown that all
known communication problems $\cP$ with $\Rand(\cP) = O(1)$ can be computed by
queries to matrices of constant support-rank. This suggests the following new
communication complexity classes:

\begin{description}[leftmargin=!,labelwidth=1.4cm,labelsep=0.5em]
    \item[$\SUPP$:]
    The class of communication problems $\cP$ with constant support-rank,
    $\suprank(\cP) = O(1)$.
    \item[$\coSUPP$:]
    The class of communication problems $\cP$ whose negation is in $\SUPP$.
    \item[$\P^{\SUPP}$:]
    The class of communication problems $\cP$ which can be computed by
    a constant-cost deterministic protocol with access to an oracle $\cQ \in \SUPP$.
\end{description}

\noindent
We note that several recent works \cite{CLV19,CHHS23,CHHNPS25,Tom25} have implicitly studied problems
in $\SUPP$: the \textsc{Integer Inner Product} functions, which belong to $\SUPP$ by definition (the
problem asks to decide if integer vectors $u, v \in [-M,M]^d$ are orthogonal in fixed dimension
$d$). For example, \cite{CLV19} proved that these problems form an infinite hierarchy and have
efficient (but still super-constant) randomized protocols.

In terms of these complexity classes, we may state non-quantitative versions of our results as:
\begin{enumerate}
    \item For all constant $k$, $\HD{\geq k} \in \SUPP$ (\cref{thm:HD-suprank}).
    \item $\P^{\SUPP} \subseteq \UPP_0$, the class of problems with constant sign-rank,
    \ie constant unbounded-error communication cost (\cref{lemma:supp-to-sign})
    \item Therefore for all constant $k$, $\P^{\HD{k}} \in \UPP_0$.
    \item Constant-order rank problems, including all $\BPP_0$ (constant-cost bounded-error randomized communication) problems from \cite{FGHH25}, are in $\P^{\SUPP} \subseteq \UPP_0$ (\cref{section:rank-problems}).
\end{enumerate}
\noindent
These new classes might informally be described as ``one-sided'' versions of $\UPP_0$.
By analogy, one might consider similar ``one-sided'' versions of constant margin,
which we call ``support margin'':

\begin{description}[leftmargin=!,labelwidth=1.4cm,labelsep=0.5em]
    \item[$\SMAR$:]
    The class of communication problems $\cP$ for which there exists a constant $\gamma > 0$
    such that for all $P \in \cP$, there exist unit vectors $u_1, \dotsc, u_N, v_1, \dotsc, v_N$
    (in arbitrary dimension) such that
    \begin{equation}
    \label{eq:def-smar}
    \begin{aligned}
        P(i,j) = 0 &\implies \inn{u_i, v_j} = 0, \\
        P(i,j) = 1 &\implies |\inn{u_i, v_j}| > \gamma .
    \end{aligned}
    \end{equation}
    \item[$\coSMAR$:]
    The class of communication problems $\cP$ whose negation is in $\SMAR$.
\end{description}
\noindent
It is not immediately obvious whether these classes are as interesting as $\SUPP$,
but we will see below that they are non-trivial classes which contain interesting problems,
in particular the problems of constant $\gamma_2$-norm.

\subsection{Relations Between Classes}

\begin{figure}[H]
    \begin{center}
        \input{pics/hierarchy.tex}
    \end{center}
	\vspace{-5mm}
    \caption{Hierarchy of constant-cost communication classes (with typical subscript $0$ dropped). New classes are shaded, and the relations we discuss in this paper are highlighted in color.}
    \label{fig:hierarchy}
\end{figure}

Let us discuss how these classes fit within the relevant landscape of
constant-cost communication complexity classes, as described in \cite{HH24}.
Communication complexity classes are sometimes written with the superscript
\emph{cc}, \eg $\BPP^\emph{cc}$, but we drop this superscript for simplicity.
The emerging convention is to write constant-cost communication classes with the
subscript $0$, \eg $\BPP_0$, by analogy to constant-depth circuit complexity
classes like $\mathsf{TC}^0$ or $\mathsf{AC}^0$, but we will also drop this subscript
since we are only concerned with the constant-cost classes. Some standard classes
are as follows, and are shown with their relations in \cref{fig:hierarchy}.

\newcommand{\ROne}{\mathsf{R}_1}
\begin{description}[leftmargin=!,labelwidth=1.4cm,labelsep=0.5em]
    \item[$\BPP$:]
    The class of communication problems $\cP$ with $\Rand(\cP)
    = O(1)$, \ie problems with constant-cost randomized public-coin
    bounded-error protocols, or equivalently families of boolean matrices with
    constant margin.
    \item[$\RP$:]
    The class of communication problems $\cP$ with $\ROne(\cP) = O(1)$,
    where $\ROne(P)$ denotes the optimal cost of a public-coin randomized protocol
    with bounded \emph{one-sided error} (\ie its output is correct with probability 1
    on inputs $x,y$ with $P(x,y) = 0$).
    \item[$\UPP$:]
    The class of communication problems $\cP$ with $\U(\cP) =
    O(1)$, \ie problems with constant-cost randomized private-coin
    unbounded-error protocols, or equivalently families of boolean matrices
    with constant sign-rank.
    \item[$\P$:]
    The class of communication problems $\cP$ with $\D(\cP) = O(1)$,
    \ie problems with constant-cost \emph{deterministic} protocols.
    \item[$\P^\EQ$:]
    The class of communication problems $\cP$ with $\D^\EQ(\cP) = O(1)$,
    \ie problems with constant-cost deterministic protocols that have access to
    an \textsc{Equality} oracle.
    \item[$\P^{\RP}$:]
    The class of communication problems $\cP$ with $\D^\cQ(\cP) = O(1)$
    for some $\cQ \in \RP$.
    \item[$\Gamma_2$:]
    The class of communication problems $\cP$ with constant $\gamma_2$-norm
    (\cref{def:gamma-2}).
\end{description}

Below, we establish some relationships between classes to fill in \cref{fig:hierarchy}.

\subsubsection{The Classes $\SUPP$, $\coSUPP$, and $\UPP$}

\begin{boxproposition}\label{prop:psup-nin-upp}
    $\P^{\SUPP} \subsetneq \UPP$.
\end{boxproposition}
The proof follows a similar Ramsey-theoretic strategy as in \cite{FHHH24}.
\begin{proof}
We show that the \textsc{Greater-Than} problem (which has sign-rank 2 as in
\cref{fig:identity}) does not belong to the class $\P^{\SUPP}$. Suppose for the sake of contradiction
that there exist constants $q, s$ such that for all $N$, the $N \times N$ \textsc{Greater-Than}
matrix $\GEQ_N \in \zo^{N \times N}$ can be written as
\[
    \GEQ_N = \Gamma( Q_1, \dotsc, Q_q )
\]
where $\Gamma \colon \zo^q \to \zo$ is applied entry-wise to the matrices $Q_i
\in \zo^{N \times N}$, and each $Q_i$ satisfies $\suprank(Q_i) \leq s$. Consider an
auxiliary complete graph $G$ on vertices $[N]$ where each edge $\{x,y\} \in
\binom{[N]}{2}$ with $x < y$ is assigned the color
\[
    \mathsf{col}(\{x,y\}) \define \left( Q_i(x,y), Q_i(y,x) \right)_{i \in [q]} .
\]
In other words, the color of $\{x,y\}$ is the vector in $\zo^{2q}$, made up of entries in
$Q_i$ for row $x$ and column $y$, as well as row $y$ and column $x$. Ramsey's
theorem guarantees that for any $n \in \bN$ there exists sufficiently large $N$
such that there is a set $T \subseteq [N]$ of size $|T| = n$ such that all edges
$\{x,y\} \in \binom{T}{2}$ have the same color. Therefore, all $x,y \in T$ with
$x < y$ share the same matrix entries $Q_i(x,y)$ and $Q_i(y,x)$. In particular,
for each $i \in [q]$ there are $b_i, b'_i \in \zo$ such that
\[
    \forall x,y \in T \text{ such that } x < y \;\colon\qquad
    Q_i(x,y) = b_i \text{ and } Q_i(y,x) = b'_i .
\]
We argue that there must exist some $i \in [q]$ such that $b_i \neq b'_i$. If this were
not the case, then for every $x,y \in T$ with $x < y$, we have
\[
    \GEQ_N(x,y) = \Gamma(Q_1(x,y), \dotsc, Q_q(x,y))
    = \Gamma(Q_1(y,x), \dotsc, Q_q(y,x)) = \GEQ_N(y,x) ,
\]
a contradiction. Therefore, we have $Q_i$ such that
\[
Q_i(x,y) = \begin{cases}
    b_i &\text{ if } x < y \\
    \neg b_i &\text{ if } x > y .
\end{cases}
\]
Then, $Q_i$ contains an $n \times n$ submatrix on the rows and columns in $T$,
which has all 1s in the upper triangle and all 0s in
the lower triangle, or vice versa. This submatrix therefore has support-rank at least $n-1$.
Since this can be found for all values of $n$, we reach a contradiction.
\end{proof}

\begin{boxproposition}\label{prop:supp-cap-cosupp}
$\SUPP\ \cap\ \coSUPP = \P$.
\end{boxproposition}
\begin{proof}
It is well-known that the class $\P$ consists of exactly the communication
problems with constant rank. Then, $\P \subseteq \SUPP \cap \coSUPP$, so we must
only show that $\SUPP \cap \coSUPP \subseteq \P$. Let $\cP$ be any communication
problem in $\SUPP \cap \coSUPP$. There exists a constant $s$ such that
for all $P \in \cP$, $\suprank(P) \leq s$ and $\suprank(\neg P) \leq s$. Lemma 3.6 of \cite{HHH23} shows that for each $N$ there is some $r$ such that
every boolean matrix $M \in \zo^{N \times N}$ with $\rank(M) \geq r$ contains an
$N \times N$ submatrix isomorphic to one of
\begin{equation}
\label{eq:hhh-forbidden}
    \EQ^N,\ \neg \EQ^N,\ \GEQ^N,\ \neg \GEQ^N .
\end{equation}
Suppose for the sake of contradiction that $\cP \notin \P$, so that for every
$r$ there exists $P \in \cP$ with $\rank(P) \geq r$. Then for every $N > s$
there exists $P \in \cP$ containing one of the matrices in
\cref{eq:hhh-forbidden} as a submatrix. But for each of these matrices, either it or its
complement has support-rank at least $N > s$, contradiction.
\end{proof}

\subsubsection{The Classes $\SMAR$, $\coSMAR$, and $\Gamma_2$}

It is now necessary to define the $\gamma_2$-norm and the class $\Gamma_2$.
\begin{definition}[$\gamma_2$-Norm and $\Gamma_2$]
\label{def:gamma-2}
Let $M \in \bR^{N \times N}$. The $\gamma_2$-norm is defined as
\[
    \gamma_2(M) \define \min_{UV = M} \|U\|_\text{row} \|V\|_\text{col} ,
\]
where the minimum is over matrices $U,V$ with $UV = M$, and $\|U\|_\text{row}$
is the maximum $\ell_2$-norm of any row of $U$, while $\|V\|_\text{col}$
is the maximum $\ell_2$-norm of any column of $V$. In other words, $\gamma_2(M)$ is the smallest
$\lambda > 0$ for which there exist real vectors $u_i, v_j$ satisfying $\|u_i\|_2, \|v_j\|_2 \leq \lambda$ and
\[
    \forall i,j \in [N] \;\colon\qquad M(i,j) = \inn{u_i, v_j}
\]
We will write $\Gamma_2$ for the class
of all communication problems $\cP$ with $\gamma_2(P) \leq \lambda$ for all
$P \in \cP$, where $\lambda$ is some constant only depending on $\cP$.
\end{definition}
In \cref{fig:hierarchy} we have shown that $\P^\EQ \subseteq \Gamma_2$, which was
proved by \cite{HHH22}. We will require the fact that $\Gamma_2$ is closed under negation:

\begin{fact}
    If $\cP \in \Gamma_2$ then $\neg \cP \in \Gamma_2$.
\end{fact}
\begin{proof}
    Fix a constant $\lambda > 0$ such that $\gamma_2(P) \leq \lambda$ for all $P
    \in \cP$. Take any $P \in \cP$ and write $P = UV$ where $\|U\|_\text{row},
    \|V\|_\text{col} \leq \lambda$. Now, $\neg P = J-UV$, where $J$ is the all-1s matrix.
    Writing $u_i$ for the $i^{th}$ row of $U$ and $v_j$ for the $j^{th}$ column of $V$, we have
    \[
        \forall i,j \;\colon\qquad \neg P(i,j) = 1 - \inn{u_i, v_j}
        = \inn{(1, u_i), (1, - v_j)} ,
    \]
    where $\|(1,u_i)\|_2^2, \|(1, -v_j\|_2^2 \leq 1 + \lambda^2$. We conclude that
    $\gamma_2(\neg P) \leq \sqrt{1+\lambda^2}$ for all $P \in \cP$.
\end{proof}

\begin{boxproposition}\label{prop:gamma-in-smar}
    $\Gamma_2 \subseteq \SMAR \cap \coSMAR$.
\end{boxproposition}
\begin{proof}
Since $\Gamma_2$ is closed under negation, it suffices to show $\Gamma_2
\subseteq \SMAR$. Let $\cP \in \Gamma_2$ and let $\lambda > 0$ be a constant
with $\gamma_2(P) \leq \lambda$ for all $P \in \cP$. For any $P \in \cP$,
write $P = UV$ with $\|U\|_\text{row}, \|V\|_\text{col} \leq \lambda$,
and write $u_i$ and $v_j$ for the $i^{th}$ row of $U$ and $j^{th}$ column of $V$,
respectively. Then, we have
\[
    P(i,j) = 0 \implies \frac{\inn{u_i, v_j}}{\|u_i\|_2 \|v_j\|_2 } = 0 \qquad\text{and}\qquad
    P(i,j) = 1 \implies \frac{ |\inn{u_i, v_j}| }{\|u_i\|_2 \|v_j\|_2 }
            \geq \frac{1}{\lambda^2}.
\]
Thus, the normalized vectors witness that $\cP \in \SMAR$ with constant $\lambda^{-2}$.
\end{proof}

\begin{boxproposition} \label{prop:rp-in-smar}
    $\RP \subseteq \SMAR$ and $\coRP \subseteq \coSMAR$.
\end{boxproposition}
\begin{proof}
It suffices to show that $\RP \subseteq \SMAR$. Let $\cP \in \RP$, so that for some
constant $c > 0$, every $P \in \cP$ has a one-sided error randomized protocol
with cost $c$. We may assume without loss of generality that the protocol
succeeds with probability at least $1/2$, and also that the protocol is one-way, \ie Alice
sends a single message to Bob, who then produces the output (this assumption
holds because we are interested only in constant cost; see \eg \cite{HWZ22}).

For each random seed $r$ and inputs $x,y$, let $a_r(x) \in \zo^c$ be the message
which Alice would send given input $x$ and random seed $r$ and let $B_r(y)
\subseteq \zo^c$ be the subset of messages on which Bob would output 1. Since
the protocol has one-sided error, for random $\bm r$, we have
\begin{align*}
    P(x,y) = 0 &\implies \Pr{ a_{\bm r}(x) \in B_{\bm r}(y) } = 0, \\
    P(x,y) = 1 &\implies \Pr{ a_{\bm r}(x) \in B_{\bm r}(y) } \geq 1/2  .
\end{align*}
We may assume that $\bm r$ is drawn uniformly from a finite universe $\zo^R$.
For a subset $S \subseteq \zo^c$ of messages, let $\chi_S \in \zo^{2^c}$
be the vector indicating membership in $S$.
Now for each $x,y$, we may construct vectors $u_x, v_y \in \zo^{2^{c+r}}$ as the
concatenations
\begin{align*}
    u_x \define \left( \chi_{\{a_r(x)\}} \right)_{r \in \zo^R}  
    \qquad\text{ and }\qquad
    v_y \define \left( \chi_{B_r(y)} \right)_{r \in \zo^R}.
\end{align*}
They have the property that
\[
    \inn{u_x, v_y}
    = 2^R \cdot \Pru{\bm r}{ \ind{ a_{\bm r}(x) \in B_{\bm r}(y) } }
    \begin{cases}
        = 0 &\text{ if } P(x,y) = 0 \\
        \geq \tfrac{1}{2} 2^R &\text{ if } P(x,y) = 1 .
    \end{cases}
\]
Moreover, the vectors have $\ell_2$-norm at most $\sqrt{2^{R+c}}$, so normalizing these
vectors gives a margin of at least $2^{-c-1}$ in the case $P(x,y) = 1$.
\end{proof}

\begin{boxproposition}
    $\P^{\SMAR \cap \coSMAR} = \SMAR \cap \coSMAR$.
\end{boxproposition}
\begin{proof}
Since $\SMAR \cap \coSMAR$ is closed under negations, it suffices to show that
$\SMAR$ is closed under OR, so that $\coSMAR$ is closed under AND and their
intersection is closed under all boolean operations. We now let $\cP, \cQ \in \SMAR$ so
that there exists a constant $\gamma > 0$ such that all $\cP$ and $\cQ$ satisfy
\cref{eq:def-smar} with constant $\gamma$. For any $N \times N$ matrices $P \in
\cP_N$ and $Q \in \cQ_N$ and each $i,j \in [N]$, let $u_i, v_j$ be the unit vectors
witnessing \cref{eq:def-smar} for $P$, and define $u'_i, v'_j$ similarly for $Q$.

The normalized concatenations $\left(\tfrac{\sqrt \gamma}{\sqrt 2} u_i, u'_i\right)$ and
$\left(\tfrac{\sqrt{\gamma}}{\sqrt 2}v_j, v'_j\right)$, both with $\ell_2$-norm $\sqrt{1 + \gamma/2}$, satisfy
\begin{align*}
    (P \vee Q)(i,j) = 0 &\implies \inn{\left(\sqrt{\tfrac{\gamma}{2}} u_i, u'_i\right), \left( \sqrt{\tfrac{\gamma}{2}} v_j, v'_j\right) } = 0, \\
    (P \vee Q)(i,j) = 1 &\implies
    \left|\inn{\left(\sqrt{\tfrac{\gamma}{2}} u_i, u'_i\right), \left( \sqrt{\tfrac{\gamma}{2}} v_j, v'_j\right) }\right| =
    \left| \tfrac{\gamma}{2} \inn{u_i, v_j} + \inn{u'_i, v'_j} \right| \geq \min\!\left( \frac{\gamma^2}{2}, \frac{\gamma(1-\gamma)}{2}\right),
\end{align*}
where the last inequality holds because $|\inn{u_i, v_j}| \geq \gamma$ or
$|\inn{u'_i, v'_j}| \geq \gamma$. Therefore, $\cP \vee \cQ \in \SMAR$.
\end{proof}

\subsection{Open Problems}\label{sec:open-problems}

We have added some classes to the hierarchy in \cref{fig:hierarchy}, which
suggests some new open problems. Recall that \cref{conj:intro-main} asks to
separate $\BPP_0$ from $\UPP_0$ (\ie show $\BPP_0 \setminus \UPP_0 \neq
\emptyset$). The introduction of support-rank suggests an intermediate problem:

\begin{question}
    Is $\BPP_0 \setminus \P_0^{\SUPP} \neq \emptyset$? (Is there a problem with
    constant bounded-error randomized cost, which cannot be reduced to any
    problem of constant support-rank?)
\end{question}

Our results show that all examples in $\BPP_0$ known up to \cite{FGHH25} are
also contained in $\P_0^{\SUPP}$. This leaves the recent examples of \cite{SS24}
as promising candidates for proving this separation.

Another question suggested by new classes in \cref{fig:hierarchy} is whether $\SMAR \cap \coSMAR = \Gamma_2$. We refer to \cite{HH24} for many other open problems regarding the complexity classes
in \cref{fig:hierarchy}, and more. Let us repeat two of the most interesting, originally from \cite{HHH22}:

\begin{question}
    Is  $\BPP_0 = \P_0^{\RP}$?
\end{question}

\begin{question}
    Is $\P_0^\EQ = \Gamma_2$?
\end{question}

\medskip
\subsection*{Acknowledgments}
We thank Arkadev Chattopadhyay for pointing out the reference~\cite{HP2010} and
Kaave Hosseini for discussions. Thanks to Tsun-Ming Cheung and Viktor Zamaraev
for pointing out several typos and a mistake in the proof of
\cref{thm:rank-problem-reduction}. All authors are supported by the Swiss State
Secretariat for Education, Research, and Innovation (SERI) under contract number
MB22.00026.

\medskip

% -----------------------------------------------
\DeclareUrlCommand{\Doi}{\urlstyle{sf}}
\renewcommand{\path}[1]{\small\Doi{#1}}
\renewcommand{\url}[1]{\href{#1}{\small\Doi{#1}}}
\bibliographystyle{alphaurl}
\bibliography{references.bib}

\end{document}

%% file: header.tex
\usepackage[letterpaper, portrait, margin=1in]{geometry}
\usepackage[colorlinks=true,linkcolor=blue,citecolor=ForestGreen]{hyperref}
\usepackage{algorithm}
\usepackage[noend]{algpseudocode}
\usepackage{url}
\usepackage{amsmath,amssymb,amsthm}
\usepackage[dvipsnames,svgnames,table]{xcolor}
\usepackage{tikz}
\usepackage{tikz-3dplot}
\usepackage{pics/tikz-3dplot-circleofsphere}
\usepackage{listofitems}
\usepackage{thmtools,thm-restate}
\usepackage[noabbrev,capitalise,nameinlink]{cleveref}
\usepackage{mathtools}
\usepackage{xspace}
\usepackage{verbatim}
\usepackage{mathrsfs}
\usepackage{pgf}
\usepackage{todo}
\usepackage{tabularx}
\usepackage{enumitem}
\usepackage{derivative}
\usepackage{bm}
\usepackage{multirow}
\usepackage{diagbox}
\usepackage{nicematrix}
\usepackage[most]{tcolorbox}
\usepackage{caption}

\usetikzlibrary{
    babel,
    arrows,
	calc,
	patterns,
	intersections,
    backgrounds,
    trees,
    mindmap,
    fadings,
	decorations.pathreplacing,
	decorations.pathmorphing,
	decorations.text,
	decorations.markings,
    scopes,
	shapes,
	positioning
}

\usepackage{dsfont}

\newcommand*\ie{i.\kern.1em e.,\ }
\newcommand*\eg{e.\kern.1em g.,\ }
\newcommand*\cf{c.\kern.1em f.\ }
\newcommand*\almev{a.\kern.1em e.\ }

\definecolor{ama-iro}{RGB}{0, 158, 243.0}
\definecolor{fuyu-gaki}{RGB}{251, 74, 52}
\definecolor{momiji}{RGB}{245, 70, 111}
\definecolor{hotaru-bi}{RGB}{229,221,58} % yellowish
\definecolor{kon-peki}{RGB}{1,120,217}
\definecolor{shin-kai}{RGB}{77,98,152}
\definecolor{shin-ryoku}{RGB}{1,145,97}
\definecolor{yama-budo}{RGB}{171,14,122}

\definecolor{citecolor}{HTML}{5E9100}

\definecolor{theoremcolor}{HTML}{FFE1D9}
\definecolor{resultcolor}{HTML}{FFE1D9}
\definecolor{constraintcolor}{HTML}{BBD49B}
\definecolor{goalcolor}{HTML}{CAE4A7}

\definecolor{remarkcolor}{HTML}{E3EEC7}

\definecolor{definitioncolor}{HTML}{FCDFBE}

\definecolor{examplecolor}{HTML}{F9E9D9}
\definecolor{questioncolor}{HTML}{DDE8EB}
\definecolor{conjecturecolor}{HTML}{DDE8EB}

\definecolor{captioncolor}{RGB}{128,128,128}

\hypersetup{
  colorlinks=true,
  linkcolor=momiji,
  filecolor=ama-iro,
  urlcolor=momiji,
  citecolor=citecolor
}

\definecolor{captioncolor}{RGB}{128,128,128}

\theoremstyle{plain}
\newtheorem{theorem}{Theorem}[section]
\newtheorem{lemma}[theorem]{Lemma}
\newtheorem{fact}[theorem]{Fact}
\newtheorem{proposition}[theorem]{Proposition}

\newtheorem{corollary}[theorem]{Corollary}
\newtheorem{question}[theorem]{Open Problem}
\newtheorem{question*}{Question}
\newtheorem{conjecture}[theorem]{Conjecture}

\newtheorem{observation}[theorem]{Observation}
\newtheorem{assumption}[theorem]{Assumption}

\crefname{claim}{Claim}{Claims}
\crefname{fact}{Fact}{Facts}

\theoremstyle{definition}

\newtheorem{definition}[theorem]{Definition}
\newtheorem{remark}[theorem]{Remark}
\newtheorem{example}[theorem]{Example}
\newtheorem{goal}[theorem]{Goal}
\newtheorem{goal*}{Goal}

\theoremstyle{plain}

\newenvironment{rtheorem}[1]{\renewcommand\thetheorem{\ref*{#1}}\begin{theorem}[Restated]}{\end{theorem}}

\newenvironment{rlemma}[1]{\renewcommand\thelemma{\ref*{#1}}\begin{lemma}[Restated]}{\end{lemma}}

\newenvironment{rdefinition}[1]{\renewcommand\thedefinition{\ref*{#1}}\begin{definition}[Restated]}{\end{definition}}

\newenvironment{boxtheorem}{\begin{theorem}}{\end{theorem}}
\tcolorboxenvironment{boxtheorem}{colback=theoremcolor, colframe=white,
    colbacktitle=theoremcolor, coltitle=theoremcolor}

\tcolorboxenvironment{rboxtheorem}{colback=theoremcolor, colframe=white,
	colbacktitle=theoremcolor, coltitle=theoremcolor}

\newenvironment{boxlemma}{\begin{lemma}}{\end{lemma}}
\tcolorboxenvironment{boxlemma}{colback=resultcolor, colframe=white,
    colbacktitle=resultcolor, coltitle=resultcolor}

\tcolorboxenvironment{rboxlemma}{colback=resultcolor, colframe=white,
	colbacktitle=resultcolor, coltitle=resultcolor}

\newenvironment{boxproposition}{\begin{proposition}}{\end{proposition}}
\tcolorboxenvironment{boxproposition}{colback=resultcolor, colframe=white,
    colbacktitle=resultcolor, coltitle=resultcolor}

\tcolorboxenvironment{boxcorollary}{colback=resultcolor, colframe=white,
    colbacktitle=resultcolor, coltitle=resultcolor}

\tcolorboxenvironment{boxobservation}{colback=resultcolor, colframe=white,
    colbacktitle=resultcolor, coltitle=resultcolor}

\newenvironment{boxquestion}{\begin{question}}{\end{question}}
\tcolorboxenvironment{boxquestion}{colback=questioncolor, colframe=white,
    colbacktitle=questioncolor, coltitle=shin-kai}
\newenvironment{boxquestion*}{\begin{question*}}{\end{question*}}
\tcolorboxenvironment{boxquestion*}{colback=questioncolor, colframe=white,
    colbacktitle=questioncolor, coltitle=questioncolor}

\newenvironment{boxdefinition}{\begin{definition}}{\end{definition}}
\tcolorboxenvironment{boxdefinition}{colback=definitioncolor, colframe=white,
    colbacktitle=definitioncolor, coltitle=definitioncolor}

\tcolorboxenvironment{rboxdefinition}{colback=definitioncolor, colframe=white,
	colbacktitle=definitioncolor, coltitle=definitioncolor}

\tcolorboxenvironment{boxassumption}{colback=definitioncolor, colframe=white,
    colbacktitle=definitioncolor, coltitle=definitioncolor}

\newenvironment{boxexample}{\begin{example}}{\end{example}}
\tcolorboxenvironment{boxexample}{colback=examplecolor, colframe=white,
    colbacktitle=examplecolor, coltitle=examplecolor}

\newtheorem{exercise}[theorem]{Exercise}

\tcolorboxenvironment{boxexercise}{colback=exercisecolor, colframe=white,
    colbacktitle=exercisecolor, coltitle=exercisecolor}

\tcolorboxenvironment{boxgoal}{colback=goalcolor, colframe=white,
    colbacktitle=goalcolor, coltitle=goalcolor}

\newenvironment{boxgoal*}{\begin{goal*}}{\end{goal*}}
\tcolorboxenvironment{boxgoal*}{colback=goalcolor, colframe=white,
    colbacktitle=goalcolor, coltitle=goalcolor}

\tcolorboxenvironment{boxremark}{colback=remarkcolor, colframe=white,
    colbacktitle=remarkcolor, coltitle=remarkcolor}

\newenvironment{boxconjecture}{\begin{conjecture}}{\end{conjecture}}
\tcolorboxenvironment{boxconjecture}{colback=conjecturecolor, colframe=white,
    colbacktitle=conjecturecolor, coltitle=shin-kai}
\newenvironment{boxconjecture*}{\begin{conjecture*}}{\end{conjecture*}}
\tcolorboxenvironment{boxconjecture*}{colback=conjecturecolor, colframe=white,
    colbacktitle=conjecturecolor, coltitle=conjecturecolor}

\newcommand{\ignore}[1]{}

\DeclareMathOperator{\poly}{poly}
\DeclareMathOperator{\sign}{sign}
\DeclareMathOperator{\dist}{dist}

\renewcommand{\Pr}[1]{\bP \left[ #1 \right]} % Probability
\newcommand{\Pru}[2]{\underset{ #1 }\bP \left[ #2 \right]}

\newcommand*{\define}{\coloneqq} % Mika: This is the standard symbol that the average reader is used to.

\newcommand{\inn}[1]{\langle #1 \rangle}

\newcommand{\ind}[1]{\mathds{1} \left[ #1 \right] }

\newcommand{\zo}{\{0,1\}}
\newcommand{\pmset}{\{\pm 1\}}

\newcommand{\cH}{\ensuremath{\mathcal{H}}}

\newcommand{\cM}{\ensuremath{\mathcal{M}}}

\newcommand{\cP}{\ensuremath{\mathcal{P}}}
\newcommand{\cQ}{\ensuremath{\mathcal{Q}}}

\newcommand{\bF}{\ensuremath{\mathbb{F}}}
\newcommand{\bN}{\ensuremath{\mathbb{N}}}
\newcommand{\bP}{\ensuremath{\mathbb{P}}}
\newcommand{\bR}{\ensuremath{\mathbb{R}}}

\tikzset{
    >=latex,
    graph-vert/.style 2 args = {
        draw,
        color = #1!70!black,
        circle,
        thick,
        inner sep = 0pt,
        minimum size = #2,
        fill = #1!50,
        fill opacity = 0.5
    },
    graph-vert/.default = {blue}{0.15cm},
    triangle/.style 2 args = {
        #1!70!black,
        fill = #1!30,
        fill opacity = #2
    },
    universe-rect/.pic = {
        \fill[black!10, rounded corners = 3pt] (-0.1, 0.1) rectangle (#1 + 0.1, -#1 - 0.1);
        \fill[white] (0, 0) rectangle (#1, -#1);
        \draw[step = 0.5, black] (0, 0) grid (#1, -#1);
        \draw[thick] (-0, 0) rectangle (#1, -#1);
    },
    vert/.style = {
        draw,
        thick,
        rectangle,
        rounded corners = 2pt,
    },
    semisim/.style = {
        ->,
        blue,
        dashed,
        decorate,
        decoration = {
            snake,
            amplitude = 0.5,
            segment length = 2
        }
    }
}

%% file: pics/half-eq.tex
\def\r{2}
\def\rr{1.5}
\tdplotsetmaincoords{60}{125}

\pgfdeclareradialshading[tikz@ball]{ball}{\pgfqpoint{-20bp}{20bp}}{%
    color(0bp)=(tikz@ball!0!white);
    color(17bp)=(tikz@ball!0!white);
    color(21bp)=(tikz@ball!70!black);
    color(25bp)=(black!70);
    color(30bp)=(black!70)}

\definecolor{mediumtaupe}{rgb}{0.4, 0.3, 0.28}
\def\colorlist{{"red", "blue", "green!80!black", "brown!50!orange"}}

\vspace{-1em}
\hspace{-3em}
\begin{tikzpicture}[tdplot_main_coords]
    
%    \foreach \i in {-75, -45, ..., 75}{
        %\tdplotCsDrawLatCircle[thin, black!20]{\r}{\i}
%    }

 %   \foreach \i in {0, 60, ..., 150}{
        %\tdplotCsDrawLonCircle[thin, black!10]{\r}{\i}
 %   }
%    \tdplotCsDrawPoint[black!50]{0}{0}{0}

    \node[inner sep=0pt] (ballpic) at (0.555,0.55)
    {\includegraphics[width=205pt]{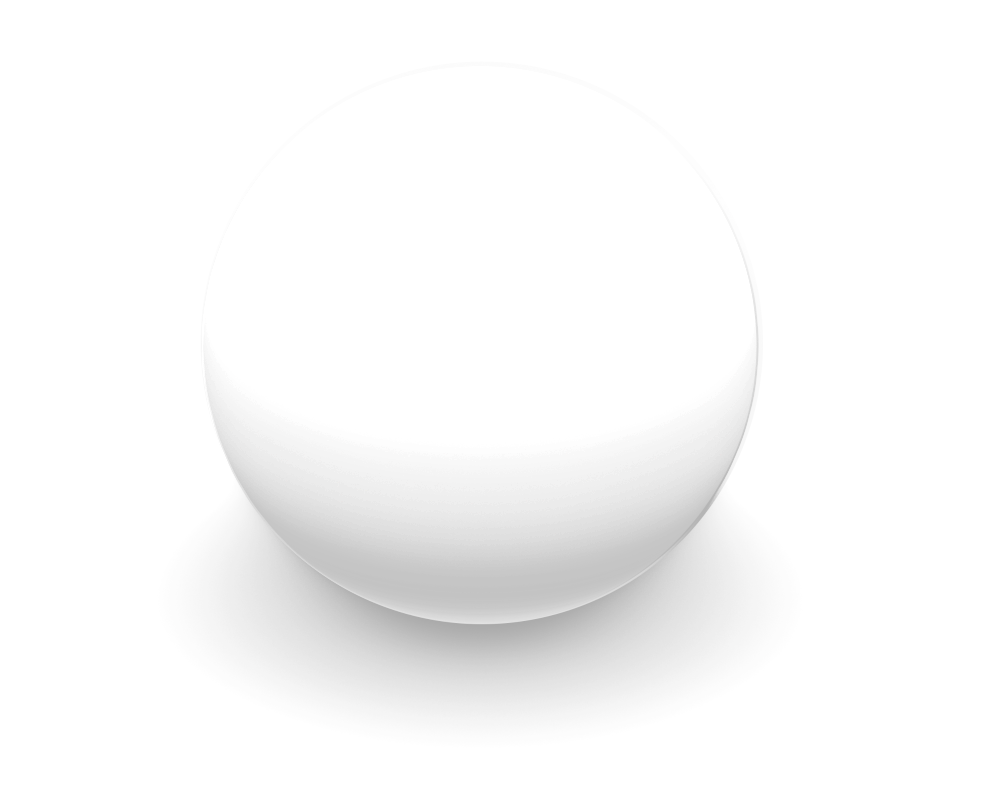}};

    \fill[black!50] (0, 0) circle (0.05);
        
    \tdplotCsDrawCircle[mediumtaupe, thin, gray,
    tdplotCsFill/.style = {opacity = 0},
    tdplotCsBack/.style = {white}]{\r}{0}{0}{45}
    \foreach \i in {0, 1, ..., 3}{        

        \pgfmathparse{\colorlist[\i]}
        \edef\col{\pgfmathresult}
        
        \tdplotCsDrawPoint[\col]{\r}{54 * \i - 40}{45.5}
        
        \tdplotCsDrawCircle[
            \col,
            thick,
            tdplotCsBack/.style = {thin, densely dotted},
            tdplotCsFill/.style = {opacity = 0.08}]{\r}{54 * \i - 40}{-50}{0}
    }
    \draw[gray, tdplot_screen_coords] (0, 0, 0) circle (\r);

    \begin{scope}[tdplot_screen_coords, shift = {(2.9, 1)}]
        \pic {universe-rect = {2}};
        \foreach \i in {0, 1, 2, 3}{
            \foreach \j in {0, 1, 2, 3}{
                \ifthenelse{\i = \j}{
                    \pgfmathparse{\colorlist[\j]}
                    \edef\col{\pgfmathresult}
                    \node at (0.25 + 0.5 * \i, -0.25 - 0.5 * \j) {\textcolor{\col!80!black}{1}};
                }{
                    \node at (0.25 + 0.5 * \i, -0.25 - 0.5 * \j) { }; % NH: Removed 0 entries for readability
                }
            }
        }
    \end{scope}
    \node[tdplot_screen_coords] at (2.5, 0) {$\cong$};

    \begin{scope}[tdplot_screen_coords, shift = {(8, 0)}]
        \draw[thick] (0, 0) circle (\rr);

       	\foreach \i in {0, 1, 2, 3}{
            \pgfmathparse{\colorlist[\i]}
            \edef\col{\pgfmathresult}

            \draw[\col, thick] (20 + 30 * \i:\rr + 0.25) -- (20 + 30 * \i + 180:\rr + 0.5);
            \fill[\col] (35 + 30 * \i:\rr) circle (2pt);
        }

        \begin{scope}[tdplot_screen_coords, shift = {(2.4, 1)}]
            \pic {universe-rect = {2}};
            \foreach \i in {0, 1, 2, 3}{
                \foreach \j in {0, 1, 2, 3}{
                    \ifthenelse{\i > \j}{
                        \node at (0.25 + 0.5 * \i, -0.25 - 0.5 * \j) { }; % NH: Removed 0 entries for readability
                    }{
                        \pgfmathparse{\colorlist[\i]}
                        \edef\col{\pgfmathresult}
                        \node at (0.25 + 0.5 * \i, -0.25 - 0.5 * \j) {\textcolor{\col!80!black}{1}};
                    }
                }
            }
        \end{scope}
        \node[tdplot_screen_coords] at (2, 0) {$\cong$};
    \end{scope}
\end{tikzpicture}
\vspace{-3.5em}

%% file: pics/khd.tex
\begin{tikzpicture}[scale = 0.5]
	
	\node at (-2.7, 0) {$\Diag(x-y)\,=$};
	
	\foreach \i in {0, 1,...,5}{
		\fill[green!20] (\i ,2 - \i) rectangle (\i+1, 3 -\i);
	}
	
	\draw[gray, thin] (0,-3) grid (6, 3);
	\draw[thick] (0,-3) rectangle (6,3);
	\draw [decorate,decoration={brace, amplitude=5pt, raise=1mm}]
	(0,3) -- (6,3) node[midway,yshift=6mm,]{$n$};

	\draw[->] (6.5,0) -- (10,0) node[above, midway, yshift=1mm] {compress} node[below, midway, yshift=-1mm] {$\Pi$};;
	\tikzset{shift={(10.5, 0)}}
	
	\fill[blue!20] (0,-2) rectangle (4,2);
	\draw[gray, thin] (0,-2) grid (4,2);
	\draw[thick] (0,-2) rectangle (4,2);
	\draw [decorate,decoration={brace, amplitude=5pt, raise=1mm}]
	(0,2) -- (4,2) node[midway,yshift=6mm,]{$k$};

	\draw[->] (4.5,0) -- (7.5,0) node[above, midway, yshift=1mm] {test} node[below, midway, yshift=-1mm] {det};	
	\node at (10.1,0.1) {Poly$(x,y)\, \stackrel{?}{=}\,  0$};

\end{tikzpicture}

%% file: pics/detsum.tex
\begin{tikzpicture}[scale = 0.65]
	\fill[blue!20] (1,5) rectangle (4,2);
	
	\fill[green!20] (0,6) rectangle (1,5);
	\fill[green!20] (4,2) rectangle (6,0);
	\fill[green!20] (0,2) rectangle (1,0);
	\fill[green!20] (4,6) rectangle (6,5);
	
	\draw[step = 1, black!40, thin] (0,0) grid (6,6);
	\draw[thick] (0,0) rectangle (6,6);
	
	\node at (0.5, 0.5) [color = citecolor] {$B_{61}$};
	\node at (1.5, 3.5) [color = blue] {$A_{32}$};
	\node at (2.5, 4.5) [color = blue] {$A_{23}$};
	\node at (3.5, 2.5) [color = blue] {$A_{44}$};
	\node at (4.5, 5.5) [color = citecolor] {$B_{15}$};
	\node at (5.5, 1.5) [color = citecolor] {$B_{56}$};
	
	\draw [decorate,decoration={brace, amplitude=5pt, mirror, raise=1mm}]
	(0,5) -- (0,2) node[midway,xshift=-6mm]{$\alpha$};
	
	\draw [decorate,decoration={brace, amplitude=5pt, raise=1mm}]
	(1,6) -- (4,6) node[midway,yshift=6mm]{$\beta$};

	\tikzset{shift={(10, 0)}}
	
	\fill[blue!20] (1,5) rectangle (4,2);
	
	\fill[green!20] (0,6) rectangle (1,5);
	\fill[green!20] (4,2) rectangle (6,0);
	\fill[green!20] (0,2) rectangle (1,0);
	\fill[green!20] (4,6) rectangle (6,5);
	
	\draw[step = 1, black!40, thin] (0,0) grid (6,6);
	\draw[thick] (0,0) rectangle (6,6);
	
	\node at (0.5, 0.5) [color = citecolor] {$B_{61}$};
	\node at (1.5, 2.5) [color = blue] {$A_{42}$};
	\node at (2.5, 3.5) [color = blue] {$A_{33}$};
	\node at (3.5, 4.5) [color = blue] {$A_{24}$};
	\node at (4.5, 5.5) [color = citecolor] {$B_{15}$};
	\node at (5.5, 1.5) [color = citecolor] {$B_{56}$};
	
	\draw [decorate,decoration={brace, amplitude=5pt, mirror, raise=1mm}]
	(0,5) -- (0,2) node[midway,xshift=-6mm]{$\alpha$};
	
	\draw [decorate,decoration={brace, amplitude=5pt, raise=1mm}]
	(1,6) -- (4,6) node[midway,yshift=6mm]{$\beta$};
	
\end{tikzpicture}

%% file: pics/rank-reduction.tex
\begin{tikzpicture}[scale = 0.7]
	
%	\fill[blue!20] (0,4) rectangle (1,3);
%	\fill[blue!20] (1,3) rectangle (2,2);
%	\fill[blue!10] (2,2) rectangle (3,1);
%	\fill[blue!20] (3,1) rectangle (4,0);
%	
%	%	\draw[step = 1, black!20, thin] (0,0) grid (4,4);
%	\draw[thick] (0,0) rectangle (4,4);
%	
%	\node at (0.5, 3.5) {$\Psi_1$};
%	\node at (1.5, 2.5) {$\Psi_2$};
%	\node at (2.5, 1.6) {$\ddots$};
%	\node at (3.5, 0.5) {$\Psi_q$};
%	
	\node at (-1, -1.15) {$A =$};
	
%	\tikzset{shift={(6,0)}}

	%	\draw[step = 1, black!20, thin] (0,0) grid (4,4);
	
	\fill[blue!20] (0,3.90) rectangle (1,3.15);
	\node at (0.5, 3.5) {$A_1$};
	
	\fill[blue!20] (1,3.15) rectangle (2,2.40);
	\node at (1.5, 2.75) {$A_2$};
	\fill[blue!20] (2,2.40) rectangle (3,1.65);
	\node at (2.5, 2) {$A_2$};
	\fill[blue!20] (3,1.65) rectangle (4,0.90);
	\node at (3.5, 1.25) {$\cdots$};
	\fill[blue!20] (4,0.90) rectangle (5,0.15);
	\node at (4.5, 0.50) {$A_2$};

        \draw[dashed] (1,3.15) rectangle (5, 0.15);
	
	\draw [decorate,decoration={brace, amplitude=5pt, mirror, raise=1mm}]
	(1.1,0) -- (4.9,0) node[midway,yshift=-6mm]{$k+1$ copies};
	
	\draw[dashed] (5,0.15) rectangle (10,-1.35);
	\node at (7.5, -0.6) {$\cdots$};
	
	\fill[blue!20] (10,-1.35) rectangle (11,-2.10);
	\node at (10.5, -1.75) {$A_q$};
    
	\fill[blue!20] (11,-2.10) rectangle (12,-2.85);
	\node at (11.5, -2.50) {$A_q$};
    
	\fill[blue!20] (12,-2.85) rectangle (13,-3.60);
	\node at (12.5, -3.25) {$A_q$};
    
	\fill[blue!20] (13,-3.60) rectangle (14,-4.35);
	\node at (13.5, -4.00) {$\cdots$};
    
	\fill[blue!20] (14,-4.35) rectangle (15,-5.10);
	\node at (14.5, -4.75) {$A_q$};
    
	\fill[blue!20] (15,-5.10) rectangle (16,-5.85);
	\node at (15.5, -5.50) {$A_q$};
    
	\draw[dashed] (10,-1.35) rectangle (16,-5.85);
	
	\draw [decorate,decoration={brace, amplitude=5pt, raise=1mm}]
	(10.1,-1.25) -- (15.9,-1.25) node[midway,yshift=6mm,xshift=2mm]{$(k+1)^{q-1}$ copies};
	
	%\draw [decorate,decoration={brace, amplitude=5pt, mirror, raise=2mm}]
	%(20,0) -- (20,4) node[midway,xshift=7mm]{$qk$};
    
	\draw[thick] (0,3.9) rectangle (16,-5.85);
	
	\draw [decorate,decoration={brace, amplitude=5pt, raise=2mm}]
	(0,4) -- (16,4) node[midway,yshift=7mm]{$(k+1)^q-1$};

\end{tikzpicture}

%% file: pics/rank-composition.tex
\begin{tikzpicture}[scale = 1.4]
		
	\fill[blue!20] (0,4) rectangle (1,3);
	\fill[blue!20] (1,3) rectangle (2,2);
	\fill[blue!10] (2,2) rectangle (3,1);
	\fill[blue!20] (3,1) rectangle (4,0);
	
	\draw[thick] (0,0) rectangle (4,4);
	
	\node at (0.5, 3.5) {$B^{(t)}_1(x_1)$};
	\node at (1.5, 2.5) {$B^{(t)}_2(x_2)$};
	\node at (2.5, 1.6) {$\ddots$};
	\node at (3.5, 0.5) {$B^{(t)}_n(x_n)$};
	
	\node at (-0.8, 2) {$B^{(t)}(x)\ =$};
	
	\draw [decorate,decoration={brace, amplitude=5pt, raise=2mm}]
	(0,4) -- (4,4) node[midway, yshift=7mm]{$nt$};
	
	\draw [decorate,decoration={brace, amplitude=5pt, raise=2mm}]
	(4,4) -- (4,0) node[midway, xshift=7mm]{$nt$};
	
\end{tikzpicture}

%% file: pics/hierarchy.tex
\begin{tikzpicture}[xscale = 1.5]
	
	\tikzset{every path/.style={->, thick}}
	
%    \node[vert] (p) at (0, 0) {$\SUPP \cap \coSUPP = \P = \RP \cap \coRP$};
	\node[vert, fill = black!10] (suppco) at (-1.8, 0) {$\SUPP \cap \coSUPP$};
	\node[circle, draw = red] at (-0.52, 0) {$=$};
	\node[vert] (p) at (0, 0) {$\P$};
	\node at (0.52, 0) {$=$};
	\node[vert] (rpco) at (1.5, 0) {$\RP \cap \coRP$};

    \node[vert, fill = black!10] (supp) at (-4.5, 2) {$\SUPP$};
    \node[vert, fill = black!10] (cosupp) at (-1.5, 2) {$\coSUPP$};
    \node[vert, fill = black!10] (psupp) at (-3, 4) {$\P^{\SUPP}$};
    \node[vert] (upp) at (-3, 8) {$\UPP$};

    \node[vert] (rp) at (1.5, 2) {$\RP$};
    \node[vert] (corp) at (4.5, 2) {$\coRP$};
    \node[vert, fill = black!10] (cosmar) at (1.5, 8) {$\SMAR$};
    \node[vert, fill = black!10] (smar) at (4.5, 8) {$\coSMAR$};
    \node[vert] (bpp) at (3, 10) {$\BPP$};

    \node[vert] (peq) at (0, 2) {$\P^{\EQ}$};
    \node[vert] (prp) at (3, 4) {$\P^{\RP}$};

    \node[vert, fill = black!10] (smarco) at (3, 6) {$\SMAR \cap \coSMAR$};

    \node[vert, fill = black!10] (gam2) at (0, 4) {$\Gamma_2$};

    \draw (suppco) -- (supp);
    \draw (suppco) -- (cosupp);
    \draw (cosupp) -- (psupp);
    \draw (supp) -- (psupp);
    \draw (psupp) -- (upp) node[circle, midway, fill = white, draw = green!80!gray] {$\neq$};

    \draw (rpco) -- (rp);
    \draw (rpco) -- (corp);
    \draw[brown] (rp) -- (cosmar);
    \draw[brown] (corp) -- (smar);
    \draw (p) -- (peq);
    \draw (smarco) -- (cosmar);
    \draw (smarco) -- (smar);
    \draw (smar) -- (bpp);
    \draw (cosmar) -- (bpp);

    \draw (peq) -- (gam2);
    \draw[blue] (gam2) to[out = 20, in = -130] (smarco);
    \draw (rp) -- (prp);
    \draw (corp) -- (prp);
    
    \draw (peq) -- (psupp);
    \draw (peq) -- (prp);
    
    \draw (prp) to[out = 30, in = -70] (bpp);
    \draw[thin, dashed] (bpp) to[out = -110, in = 150]  node[circle, midway, fill = white, draw = cyan, solid, thick] {$?$} (prp);
    
    \draw[thin, dashed] (gam2) to[out = -50, in = 60]  node[circle, midway, fill = white, draw = cyan, solid, thick] {$?$} (peq);

    \draw[thin, dashed] (gam2) -- (upp)
    	node[circle, midway, fill = white, draw = cyan, solid, thick] {$?$};

    \node at (0, 6.9) {};
    
	\node[minimum size=5pt, fill = green!70!gray, label = right:{\Cref{prop:psup-nin-upp}}] at (-4, -1) {};
	\node[minimum size=5pt, fill = red!60, label = right:{\Cref{prop:supp-cap-cosupp}}] at (-1, -1) {};
	\node[minimum size=5pt, fill = blue!60, label = right:{\Cref{prop:gamma-in-smar}}] at (2, -1) {};
	\node[minimum size=5pt, fill = brown!70, label = right:{\Cref{prop:rp-in-smar}}] at (-4, -1.8) {};
	\node[minimum size=5pt, fill = cyan!60, label = right:{The three open problems stated in \Cref{sec:open-problems}}] at (-1, -1.8) {};

\end{tikzpicture}

%% file: supp-rank.bbl
\newcommand{\etalchar}[1]{$^{#1}$}
 \providecommand{\FOCS}{Proceedings of the IEEE Symposium on Foundations of
  Computer Science (FOCS)} \providecommand{\SODA}{Proceedings of the ACM-SIAM
  Symposium on Discrete Algorithms (SODA)} \providecommand{\STOC}{Proceedings
  of the ACM SIGACT Symposium on Theory of Computing (STOC)}
  \providecommand{\ITCS}{Proceedings of the Innovations in Theoretical Computer
  Science Conference (ITCS)} \providecommand{\ICALP}{Proceedings of the
  International Colloquium on Automata, Languages, and Programming (ICALP)}
  \providecommand{\ICML}{Proceedings of the International Conference on Machine
  Learning (ICML)} \providecommand{\COLT}{Proceedings of the Conference on
  Learning Theory (COLT)} \providecommand{\AISTATS}{Proceedings of the
  International Conference on Artificial Intelligence and Statistics (AISTATS)}
  \providecommand{\TOCT}{ACM Transactions on Computation Theory (TOCT)}
  \providecommand{\RANDOM}{Approximation, Randomization, and Combinatorial
  Optimization. Algorithms and Techniques (APPROX/RANDOM)}
  \providecommand{\JACM}{Journal of the ACM (JACM)}
  \providecommand{\SIAMJOC}{SIAM Journal on Computing}
  \providecommand{\TOC}{Theory of Computing} \providecommand{\TOIT}{IEEE
  Transactions on Information Theory} \providecommand{\COLT}{Conference on
  Learning Theory (COLT)} \providecommand{\NEURIPS}{Advances in Neural
  Information Processing Systems (NeurIPS)} \providecommand{\JMLR}{Journal of
  Machine Learning Research} \providecommand{\SOSA}{Symposium on Simplicity in
  Algorithms (SOSA)} \providecommand{\FSTTCS}{IARCS Annual Conference on
  Foundations of Software Technology and Theoretical Computer Science (FSTTCS)}
\begin{thebibliography}{CHH{\etalchar{+}}25}

\bibitem[ABSZ24]{ABSZ24}
Sarosh Adenwalla, Samuel Braunfeld, John Sylvester, and Viktor Zamaraev.
\newblock Boolean combinations of graphs.
\newblock {\em arXiv preprint arXiv:2412.19551}, 2024.
\newblock \href {https://doi.org/10.48550/arXiv.2412.19551}
  {\path{doi:10.48550/arXiv.2412.19551}}.

\bibitem[ACHS24]{ACHS24}
Manasseh Ahmed, Tsun-Ming Cheung, Hamed Hatami, and Kusha Sareen.
\newblock Communication complexity and discrepancy of halfplanes.
\newblock In {\em 40th International Symposium on Computational Geometry (SoCG
  2024)}, pages 5--1. Schloss Dagstuhl--Leibniz-Zentrum f{\"u}r Informatik,
  2024.
\newblock \href {https://doi.org/10.4230/LIPIcs.SoCG.2024.5}
  {\path{doi:10.4230/LIPIcs.SoCG.2024.5}}.

\bibitem[AK14]{AK14}
Noga Alon and Andrey Kupavskii.
\newblock Two notions of unit distance graphs.
\newblock {\em Journal of Combinatorial Theory, Series A}, 125:1--17, 2014.
\newblock \href {https://doi.org/10.1016/j.jcta.2014.02.006}
  {\path{doi:10.1016/j.jcta.2014.02.006}}.

\bibitem[AN25]{AN25}
Benny Applebaum and Oded Nir.
\newblock The meta-complexity of secret sharing.
\newblock In {\em \STOC}, 2025.

\bibitem[BCZ17]{BCZ17}
Harry Buhrman, Matthias Christandl, and Jeroen Zuiddam.
\newblock Nondeterministic quantum communication complexity: the cyclic
  equality game and iterated matrix multiplication.
\newblock In {\em Proceedings of the 8th Innovations in Theoretical Computer
  Science Conference (ITCS)}. Schloss Dagstuhl, 2017.
\newblock \href {https://doi.org/10.4230/LIPICS.ITCS.2017.24}
  {\path{doi:10.4230/LIPICS.ITCS.2017.24}}.

\bibitem[BCZ18]{BCZ18}
Markus Bl{\"a}ser, Matthias Christandl, and Jeroen Zuiddam.
\newblock The border support rank of two-by-two matrix multiplication is seven.
\newblock {\em Chicago Journal of Theoretical Computer Science}, 24(1):1--16,
  2018.
\newblock \href {https://doi.org/10.4086/cjtcs.2018.005}
  {\path{doi:10.4086/cjtcs.2018.005}}.

\bibitem[BES02]{BES02}
Shai {Ben-David}, Nadav Eiron, and Hans~Ulrich Simon.
\newblock Limitations of learning via embeddings in euclidean half spaces.
\newblock {\em Journal of Machine Learning Research}, 3(Nov):441--461, 2002.

\bibitem[BFS86]{BFS86}
L{\'a}szl{\'o} Babai, Peter Frankl, and Janos Simon.
\newblock Complexity classes in communication complexity theory.
\newblock In {\em 27th Annual Symposium on Foundations of Computer Science
  (sfcs 1986)}, pages 337--347. IEEE, 1986.
\newblock \href {https://doi.org/10.1109/SFCS.1986.15}
  {\path{doi:10.1109/SFCS.1986.15}}.

\bibitem[BHH{\etalchar{+}}25]{BHHLT25}
Ari Blondal, Hamed Hatami, Pooya Hatami, Chavdar Lalov, and Sivan Tretiak.
\newblock Borsuk-ulam and replicable learning of large-margin halfspaces.
\newblock {\em arXiv preprint arXiv:2503.15294}, 2025.
\newblock \href {https://doi.org/10.48550/arXiv.2503.15294}
  {\path{doi:10.48550/arXiv.2503.15294}}.

\bibitem[BMT21]{BMT21}
Mark Bun, Nikhil~S Mande, and Justin Thaler.
\newblock Sign-rank can increase under intersection.
\newblock {\em ACM Transactions on Computation Theory (TOCT)}, 13(4):1--17,
  2021.
\newblock \href {https://doi.org/10.1145/3470863} {\path{doi:10.1145/3470863}}.

\bibitem[BNS19]{BNS19}
Amos Beimel, Kobbi Nissim, and Uri Stemmer.
\newblock Characterizing the sample complexity of pure private learners.
\newblock {\em Journal of Machine Learning Research}, 20(146):1--33, 2019.

\bibitem[BW16]{BW16}
Mark Braverman and Omri Weinstein.
\newblock A discrepancy lower bound for information complexity.
\newblock {\em Algorithmica}, 76:846--864, 2016.
\newblock \href {https://doi.org/10.1007/s00453-015-0093-8}
  {\path{doi:10.1007/s00453-015-0093-8}}.

\bibitem[CHH{\etalchar{+}}25]{CHHNPS25}
Tsun-Ming Cheung, Hamed Hatami, Kaave Hosseini, Aleksandar Nikolov, Toniann
  Pitassi, and Morgan Shirley.
\newblock A lower bound on the trace norm of boolean matrices and its
  applications.
\newblock In {\em 16th Innovations in Theoretical Computer Science Conference
  (ITCS 2025)}, pages 37--1. Schloss Dagstuhl--Leibniz-Zentrum f{\"u}r
  Informatik, 2025.
\newblock \href {https://doi.org/10.4230/LIPIcs.ITCS.2025.37}
  {\path{doi:10.4230/LIPIcs.ITCS.2025.37}}.

\bibitem[CHHS23]{CHHS23}
Tsun-Ming Cheung, Hamed Hatami, Kaave Hosseini, and Morgan Shirley.
\newblock Separation of the factorization norm and randomized communication
  complexity.
\newblock In {\em 38th Computational Complexity Conference (CCC 2023)}, pages
  1--1. Schloss Dagstuhl--Leibniz-Zentrum f{\"u}r Informatik, 2023.
\newblock \href {https://doi.org/10.4230/LIPIcs.CCC.2023.1}
  {\path{doi:10.4230/LIPIcs.CCC.2023.1}}.

\bibitem[CHZZ24]{CHZZ24}
Tsun~Ming Cheung, Hamed Hatami, Rosie Zhao, and Itai Zilberstein.
\newblock Boolean functions with small approximate spectral norm.
\newblock {\em Discrete Analysis}, 2024.
\newblock \href {https://doi.org/10.19086/da.122971}
  {\path{doi:10.19086/da.122971}}.

\bibitem[CLV19]{CLV19}
Arkadev Chattopadhyay, Shachar Lovett, and Marc Vinyals.
\newblock Equality alone does not simulate randomness.
\newblock In {\em 34th Computational Complexity Conference (CCC 2019)}, pages
  14--1. Schloss Dagstuhl--Leibniz-Zentrum f{\"u}r Informatik, 2019.
\newblock \href {https://doi.org/10.4230/LIPIcs.CCC.2019.14}
  {\path{doi:10.4230/LIPIcs.CCC.2019.14}}.

\bibitem[CU13]{CU13}
Henry Cohn and Christopher Umans.
\newblock Fast matrix multiplication using coherent configurations.
\newblock In {\em Proceedings of the 24th Symposium on Discrete Algorithms
  (SODA)}, pages 1074--1087. Society for Industrial and Applied Mathematics,
  January 2013.
\newblock \href {https://doi.org/10.1137/1.9781611973105.77}
  {\path{doi:10.1137/1.9781611973105.77}}.

\bibitem[DHP{\etalchar{+}}22]{DHPTU22}
Ben Davis, Hamed Hatami, William Pires, Ran Tao, and Hamza Usmani.
\newblock On public-coin zero-error randomized communication complexity.
\newblock {\em Information Processing Letters}, 178:106293, 2022.
\newblock \href {https://doi.org/10.1016/j.ipl.2022.106293}
  {\path{doi:10.1016/j.ipl.2022.106293}}.

\bibitem[dW03]{dWol03}
Ronald de~Wolf.
\newblock Nondeterministic quantum query and communication complexities.
\newblock {\em SIAM Journal on Computing}, 32(3):681--699, January 2003.
\newblock \href {https://doi.org/10.1137/s0097539702407345}
  {\path{doi:10.1137/s0097539702407345}}.

\bibitem[EHK22]{EHK22}
Louis Esperet, Nathaniel Harms, and Andrey Kupavskii.
\newblock Sketching distances in monotone graph classes.
\newblock In {\em Approximation, Randomization, and Combinatorial Optimization.
  Algorithms and Techniques (APPROX/RANDOM 2022)}, pages 18--1. Schloss
  Dagstuhl--Leibniz-Zentrum f{\"u}r Informatik, 2022.
\newblock \href {https://doi.org/10.4230/LIPIcs.APPROX/RANDOM.2022.18}
  {\path{doi:10.4230/LIPIcs.APPROX/RANDOM.2022.18}}.

\bibitem[EHT65]{EHT65}
Paul Erd{\H{o}}s, Frank Harary, and William Tutte.
\newblock On the dimension of a graph.
\newblock {\em Mathematika}, 12(2):118--122, December 1965.
\newblock \href {https://doi.org/10.1112/s0025579300005222}
  {\path{doi:10.1112/s0025579300005222}}.

\bibitem[FGHH25]{FGHH25}
Yuting Fang, Mika G\"o\"os, Nathaniel Harms, and Pooya Hatami.
\newblock Constant-cost communication does not reduce to {$k$-Hamming
  Distance}.
\newblock In {\em \STOC}, 2025.
\newblock \href {https://doi.org/10.48550/arXiv.2407.20204}
  {\path{doi:10.48550/arXiv.2407.20204}}.

\bibitem[FH07]{FH07}
Shaun Fallat and Leslie Hogben.
\newblock The minimum rank of symmetric matrices described by a graph: A
  survey.
\newblock {\em Linear Algebra and its Applications}, 426(2–3):558--582, 2007.
\newblock \href {https://doi.org/10.1016/j.laa.2007.05.036}
  {\path{doi:10.1016/j.laa.2007.05.036}}.

\bibitem[FHHH24]{FHHH24}
Yuting Fang, Lianna Hambardzumyan, Nathaniel Harms, and Pooya Hatami.
\newblock No complete problem for constant-cost randomized communication.
\newblock In {\em \STOC}, 2024.
\newblock \href {https://doi.org/10.48550/arXiv.2404.00812}
  {\path{doi:10.48550/arXiv.2404.00812}}.

\bibitem[FK09]{FK09}
Pierre Fraigniaud and Amos Korman.
\newblock On randomized representations of graphs using short labels.
\newblock In {\em Proceedings of the twenty-first annual symposium on
  Parallelism in algorithms and architectures}, pages 131--137, 2009.
\newblock \href {https://doi.org/10.1145/1583991.1584031}
  {\path{doi:10.1145/1583991.1584031}}.

\bibitem[FKL{\etalchar{+}}01]{FKL+01}
J{\"u}rgen Forster, Matthias Krause, Satyanarayana Lokam, Rustam Mubarakzjanov,
  Niels Schmitt, and Hans~Ulrich Simon.
\newblock Relations between communication complexity, linear arrangements, and
  computational complexity.
\newblock In {\em International Conference on Foundations of Software
  Technology and Theoretical Computer Science}, pages 171--182. Springer, 2001.
\newblock \href {https://doi.org/10.1007/3-540-45294-X_15}
  {\path{doi:10.1007/3-540-45294-X_15}}.

\bibitem[For02]{For02}
J{\"u}rgen Forster.
\newblock A linear lower bound on the unbounded error probabilistic
  communication complexity.
\newblock {\em Journal of Computer and System Sciences}, 65(4):612--625, 2002.
\newblock \href {https://doi.org/10.1016/S0022-0000(02)00019-3}
  {\path{doi:10.1016/S0022-0000(02)00019-3}}.

\bibitem[FX14]{FX14}
Vitaly Feldman and David Xiao.
\newblock Sample complexity bounds on differentially private learning via
  communication complexity.
\newblock In {\em Conference on Learning Theory}, pages 1000--1019. PMLR, 2014.
\newblock \href {https://doi.org/10.48550/arXiv.1402.6278}
  {\path{doi:10.48550/arXiv.1402.6278}}.

\bibitem[Har20]{Har20}
Nathaniel Harms.
\newblock Universal communication, universal graphs, and graph labeling.
\newblock In {\em 11th Innovations in Theoretical Computer Science Conference
  (ITCS 2020)}, volume 151, page~33. Schloss Dagstuhl--Leibniz-Zentrum fuer
  Informatik, 2020.
\newblock \href {https://doi.org/10.4230/LIPIcs.ITCS.2020.33}
  {\path{doi:10.4230/LIPIcs.ITCS.2020.33}}.

\bibitem[HH24]{HH24}
Hamed Hatami and Pooya Hatami.
\newblock Guest column: Structure in communication complexity and constant-cost
  complexity classes.
\newblock {\em ACM SIGACT News}, 55(1):67--93, 2024.
\newblock \href {https://doi.org/10.1145/3654780.3654788}
  {\path{doi:10.1145/3654780.3654788}}.

\bibitem[HHH22]{HHH22}
Lianna Hambardzumyan, Hamed Hatami, and Pooya Hatami.
\newblock A counter-example to the probabilistic universal graph conjecture via
  randomized communication complexity.
\newblock {\em Discrete Applied Mathematics}, 322:117--122, 2022.
\newblock \href {https://doi.org/10.1016/j.dam.2022.07.023}
  {\path{doi:10.1016/j.dam.2022.07.023}}.

\bibitem[HHH23]{HHH23}
Lianna Hambardzumyan, Hamed Hatami, and Pooya Hatami.
\newblock Dimension-free bounds and structural results in communication
  complexity.
\newblock {\em Israel Journal of Mathematics}, 253(2):555--616, 2023.
\newblock \href {https://doi.org/10.1007/s11856-022-2365-8}
  {\path{doi:10.1007/s11856-022-2365-8}}.

\bibitem[HHL20]{HHL20}
Hamed Hatami, Kaave Hosseini, and Shachar Lovett.
\newblock Sign rank vs discrepancy.
\newblock In {\em 35th Computational Complexity Conference (CCC 2020)}, pages
  18--1. Schloss Dagstuhl--Leibniz-Zentrum f{\"u}r Informatik, 2020.
\newblock \href {https://doi.org/10.4230/LIPIcs.CCC.2020.18}
  {\path{doi:10.4230/LIPIcs.CCC.2020.18}}.

\bibitem[HHM23]{HHM23}
Hamed Hatami, Kaave Hosseini, and Xiang Meng.
\newblock A {Borsuk-Ulam} lower bound for sign-rank and its applications.
\newblock In {\em \STOC}, pages 463--471, 2023.
\newblock \href {https://doi.org/10.1145/3564246.3585210}
  {\path{doi:10.1145/3564246.3585210}}.

\bibitem[HHP{\etalchar{+}}22]{HHPTZ22}
Hamed Hatami, Pooya Hatami, William Pires, Ran Tao, and Rosie Zhao.
\newblock Lower bound methods for sign-rank and their limitations.
\newblock In {\em \RANDOM}, pages 22--1. Schloss Dagstuhl--Leibniz-Zentrum
  f{\"u}r Informatik, 2022.
\newblock \href {https://doi.org/10.4230/LIPIcs.APPROX/RANDOM.2022.22}
  {\path{doi:10.4230/LIPIcs.APPROX/RANDOM.2022.22}}.

\bibitem[HP10]{HP2010}
Kristoffer~Arnsfelt Hansen and Vladimir Podolskii.
\newblock Exact threshold circuits.
\newblock In {\em Proceedings of the 25th Conference on Computational
  Complexity (CCC)}, pages 270--279. IEEE, 2010.
\newblock \href {https://doi.org/10.1109/ccc.2010.33}
  {\path{doi:10.1109/ccc.2010.33}}.

\bibitem[HQ17]{HQ17}
Hamed Hatami and Yingjie Qian.
\newblock The unbounded-error communication complexity of symmetric {XOR}
  functions.
\newblock {\em arXiv preprint arXiv:1704.00777}, 2017.
\newblock \href {https://doi.org/10.48550/arXiv.1704.00777}
  {\path{doi:10.48550/arXiv.1704.00777}}.

\bibitem[HR24]{HR24}
Nathaniel Harms and Artur Riazanov.
\newblock Better boosting of communication oracles, or not.
\newblock In {\em \FSTTCS}, 2024.
\newblock \href {https://doi.org/10.4230/LIPIcs.FSTTCS.2024.25}
  {\path{doi:10.4230/LIPIcs.FSTTCS.2024.25}}.

\bibitem[HSZZ06]{HSZZ06}
Wei Huang, Yaoyun Shi, Shengyu Zhang, and Yufan Zhu.
\newblock The communication complexity of the hamming distance problem.
\newblock {\em Information Processing Letters}, 99(4):149--153, 2006.
\newblock \href {https://doi.org/10.1016/j.ipl.2006.01.014}
  {\path{doi:10.1016/j.ipl.2006.01.014}}.

\bibitem[HWZ22]{HWZ22}
Nathaniel Harms, Sebastian Wild, and Viktor Zamaraev.
\newblock Randomized communication and implicit graph representations.
\newblock In {\em \STOC}, 2022.
\newblock \href {https://doi.org/10.1145/3519935.3519978}
  {\path{doi:10.1145/3519935.3519978}}.

\bibitem[HZ24]{HZ24}
Nathaniel Harms and Viktor Zamaraev.
\newblock Randomized communication and implicit representations for matrices
  and graphs of small sign-rank.
\newblock In {\em \SODA}, pages 1810--1833. SIAM, 2024.
\newblock \href {https://doi.org/10.1137/1.9781611977912.72}
  {\path{doi:10.1137/1.9781611977912.72}}.

\bibitem[LMSS07]{LMSS07}
Nati Linial, Shahar Mendelson, Gideon Schechtman, and Adi Shraibman.
\newblock Complexity measures of sign matrices.
\newblock {\em Combinatorica}, 27:439--463, 2007.
\newblock \href {https://doi.org/10.1007/s00493-007-2160-5}
  {\path{doi:10.1007/s00493-007-2160-5}}.

\bibitem[LS09]{LS09}
Nati Linial and Adi Shraibman.
\newblock Learning complexity vs communication complexity.
\newblock {\em Combinatorics, Probability and Computing}, 18(1-2):227--245,
  2009.
\newblock \href {https://doi.org/10.1017/S0963548308009656}
  {\path{doi:10.1017/S0963548308009656}}.

\bibitem[Mar90]{Mar90}
Marvin Marcus.
\newblock Determinants of sums.
\newblock {\em The College Mathematics Journal}, 21(2):130--135, 1990.
\newblock \href {https://doi.org/10.1080/07468342.1990.11973297}
  {\path{doi:10.1080/07468342.1990.11973297}}.

\bibitem[New91]{New91}
Ilan Newman.
\newblock Private vs. common random bits in communication complexity.
\newblock {\em Information Processing Letters}, 39(2):67--71, 1991.
\newblock \href {https://doi.org/10.1016/0020-0190(91)90157-d}
  {\path{doi:10.1016/0020-0190(91)90157-d}}.

\bibitem[NP24]{NP24}
Moni Naor and Eugene Pekel.
\newblock Adjacency sketches in adversarial environments.
\newblock In {\em \SODA}, pages 1067--1098. SIAM, 2024.
\newblock \href {https://doi.org/10.48550/arXiv.2111.03639}
  {\path{doi:10.48550/arXiv.2111.03639}}.

\bibitem[PS86]{PS86}
Ramamohan Paturi and Janos Simon.
\newblock Probabilistic communication complexity.
\newblock {\em Journal of Computer and System Sciences}, 33(1):106--123, 1986.
\newblock \href {https://doi.org/10.1016/0022-0000(86)90046-2}
  {\path{doi:10.1016/0022-0000(86)90046-2}}.

\bibitem[PSS23]{PSS23}
Toniann Pitassi, Morgan Shirley, and Adi Shraibman.
\newblock The strength of equality oracles in communication.
\newblock In {\em 14th Innovations in Theoretical Computer Science Conference
  (ITCS 2023)}, pages 89--1. Schloss Dagstuhl--Leibniz-Zentrum f{\"u}r
  Informatik, 2023.
\newblock \href {https://doi.org/10.4230/LIPIcs.ITCS.2023.89}
  {\path{doi:10.4230/LIPIcs.ITCS.2023.89}}.

\bibitem[Sa{\u{g}}18]{Sag18}
Mert Sa{\u{g}}lam.
\newblock Near log-convexity of measured heat in (discrete) time and
  consequences.
\newblock In {\em 2018 IEEE 59th Annual Symposium on Foundations of Computer
  Science (FOCS)}, pages 967--978. IEEE, 2018.
\newblock \href {https://doi.org/10.1109/FOCS.2018.00095}
  {\path{doi:10.1109/FOCS.2018.00095}}.

\bibitem[SS24]{SS24}
Alexander Sherstov and Andrey Storozhenko.
\newblock The communication complexity of approximating matrix rank.
\newblock In {\em IEEE Annual Symposium on Foundations of Computer Science
  (FOCS)}, pages 433--462. IEEE, 2024.
\newblock \href {https://doi.org/10.1109/FOCS61266.2024.00035}
  {\path{doi:10.1109/FOCS61266.2024.00035}}.

\bibitem[SY23]{SA23}
Srikanth Srinivasan and Amir Yehudayoff.
\newblock The discrepancy of greater-than.
\newblock {\em arXiv preprint arXiv:2309.08703}, 2023.
\newblock \href {https://doi.org/10.48550/ARXIV.2309.08703}
  {\path{doi:10.48550/ARXIV.2309.08703}}.

\bibitem[Tom25]{Tom25}
Istv{\'a}n Tomon.
\newblock Factorization norms and {Zarankiewicz} problems.
\newblock {\em arXiv preprint arXiv:2502.18429}, 2025.
\newblock \href {https://doi.org/10.48550/arXiv.2502.18429}
  {\path{doi:10.48550/arXiv.2502.18429}}.

\bibitem[Vio15]{Vio15}
Emanuele Viola.
\newblock The communication complexity of addition.
\newblock {\em Combinatorica}, 35(6):703--747, 2015.
\newblock \href {https://doi.org/10.1007/s00493-014-3078-3}
  {\path{doi:10.1007/s00493-014-3078-3}}.

\end{thebibliography}
